\documentclass[12pt]{article}

\usepackage{amsmath,amssymb}
\usepackage{authblk}
\usepackage{enumitem}  
\usepackage{hyperref}  
\usepackage[usenames, dvipsnames]{color} 
\usepackage{graphicx}     
\usepackage{subcaption}
\usepackage[utf8]{inputenc}
\usepackage{url}
 \usepackage{xcolor}
\usepackage{cite}
\usepackage{esint}
\usepackage{amsmath, amsthm}

\newtheorem*{prop}{Proposition}
  \newtheorem*{theorem}{Theorem}

\hypersetup{
    bookmarks=true,         
    unicode=false,          
    pdftoolbar=true,        
    pdfmenubar=true,        
    pdffitwindow=false,     
    pdfstartview={FitH},    
    pdftitle={My title},    
    pdfauthor={Author},     
    pdfsubject={Subject},   
    pdfcreator={Creator},   
    pdfproducer={Producer}, 
    pdfkeywords={keyword1} {key2} {key3}, 
    pdfnewwindow=true,      
    colorlinks=true,       
    linkcolor=black,           
    citecolor=black,        
    filecolor=black,      
    urlcolor=black         
} 

\def\beq{\begin{equation}}
\def\eeq{\end{equation}}

\begin{document}

\numberwithin{equation}{section}
\title{\LARGE \bf Holographic Stirling engines \\
and the route to Carnot efficiency}

\author[1]{Nikesh Lilani}   
\author[2]{Manus R. Visser\thanks{manus.visser@ru.nl}}

\affil[1]{\it{Department of Applied Mathematics and Theoretical Physics, University of Cambridge, CB3 0WA, UK}}
\affil[2]{\it{Institute for Mathematics, Astrophysics and Particle Physics, 
and Radboud Center for Natural Philosophy, Radboud University,
6525 AJ Nijmegen, The Netherlands}}

 \date{\today}
 
\maketitle

\begin{abstract}
We compute the efficiency of the reversible Stirling engine, with and without
regeneration, for a broad class of working substances including Van der Waals
fluids, quantum ideal gases (Bose and Fermi), Bose-Einstein condensates,
thermal conformal field theories (CFTs), and holographic CFTs. Regeneration
acts as an internal heat recycling mechanism that enhances efficiency by
reducing the net heat exchange with external reservoirs.
\textcolor{black}{For regenerative Stirling cycles, we identify a general
sufficient condition for attaining Carnot efficiency: at fixed conserved
charges, the fixed-volume heat capacity is independent of volume. This ensures
that, at each temperature along the isochores, the heat released during cooling
equals the heat absorbed during heating. An ideal regenerator can therefore
store and return this heat reversibly, returning to its initial thermodynamic
state with no net change in energy or entropy.}
While this condition is satisfied for classical ideal gases and Van der Waals
fluids, it is violated for quantum ideal gases and CFT working substances.
\textcolor{black}{For thermal CFT states dual to AdS-Schwarzschild and
AdS-Reissner-Nordstr\"om black holes, we obtain exact expressions for the
Stirling efficiency. In the fixed-potential ensemble, both the regenerative
and non-regenerative efficiencies approach the Carnot value in the
large-potential limit. Regeneration reduces the leading deviation from the
Carnot value, although both efficiencies have corrections of the same
asymptotic order.}

\end{abstract}

\thispagestyle{empty}

\newpage
 
{ 
\small \tableofcontents
\thispagestyle{empty}
}
 
\section{Introduction}
\label{sec:intro}

Black holes provide one of the sharpest arenas in which gravitation, thermodynamics, and quantum theory intersect. If black holes are genuine thermodynamic systems, then they may be employed, at least in principle, in heat engines. The idea that black holes can participate in heat engine  processes originates in a sequence of thought experiments, beginning with Geroch's box, that exposed deep tensions between classical gravity and the laws of thermodynamics. In these early constructions, black holes do not serve as working substances, but rather as thermal reservoirs. Only in more recent developments have black holes themselves been promoted to working substances in well-defined thermodynamic cycles, and have been given a dual   interpretation via holography. We review these two  developments, Geroch's heat engines and holographic heat engines, in turn below.\\

\noindent \textbf{Geroch's black hole heat engine.} The starting point is a thought experiment proposed by   Geroch \cite{Gerochv2} in a 1971 Princeton colloquium, in which a procedure is described that appears to violate the Kelvin  version of the second law by converting heat into work with unit efficiency. According to Bekenstein's description in \cite{Bekenstein:1972tm,Bekenstein:1973ur}, and Wald's recollection~\cite{Wald:2018wro}, Geroch considered a box of total rest energy $m$ and temperature $T_{\rm h}$, which is slowly lowered using a string from infinity toward the horizon of a Schwarzschild black hole of mass~$M$. In this setup, the working system is the external matter in the box (typically  filled with thermal radiation), while the black hole acts as a   reservoir at temperature $T_{\rm c}$ into which heat can be deposited.

As the box is lowered quasi-statically, its energy as measured at infinity is redshifted according to
\begin{equation}
E (r) = m\chi(r)\,, \qquad \chi(r) = \sqrt{1 - \frac{2M}{r}}\,.
\end{equation}
By lowering the box arbitrarily close to the horizon, this energy can be made arbitrarily small, $E(r) \to 0$ as $r \to 2M$, so the agent lowering the box extracts an amount of work that approaches $m$. At this point, the box is allowed to emit radiation into the black hole, reducing its rest energy from $m$ to $m - \Delta m$. The box is then pulled back to infinity, requiring work $m - \Delta m$. The net work gained in the entire process therefore approaches $W \to m -(m - \Delta m)= \Delta m$.

Thus, an amount of heat $Q_{\rm in} = \Delta m$ has been completely converted into work, corresponding to an efficiency
\begin{equation}
\eta = \frac{W}{Q_{\rm in}} \to 1 \, .
\end{equation}
At the same time,  as the energy $E $ delivered to the black hole vanishes as $r \to 2M$,   the black hole mass is unchanged,  hence   the black hole returns to its original state in this idealized limit. In ordinary thermodynamics, according to Carnot's theorem, a reversible heat engine operating between temperatures $T_{\rm h}$ and $T_{\rm c}$ satisfies the Carnot efficiency $\eta_{\rm Carnot} = 1 - T_{\rm c}/T_{\rm h}$. Achieving $\eta \to 1$ therefore requires $T_{\rm c} \to 0$, indicating that, at the classical level, the black hole behaves as a reservoir of effectively zero temperature.

Bekenstein \cite{Bekenstein:1972tm,Bekenstein:1973ur} recognized that the Geroch process leads to an apparent violation of the second law of thermodynamics, even if the black hole itself is assigned an entropy. His introduction of black hole entropy was in fact motivated by a different thought experiment, suggested by his supervisor Wheeler \cite{Bekenstein1980PhysicsToday}, in which a hypothetical ``Wheeler's demon'' lowers a system carrying entropy into a black hole, causing it to disappear from the exterior universe. Once the black hole settles down to a stationary state, this entropy is no longer accessible to an external observer, since the black hole is characterized only by a small set of   parameters due to the no-hair theorem in general relativity. 

Bekenstein proposed that this apparent loss is resolved if the black hole itself carries entropy, and argued that this entropy should be proportional to the horizon area,
\begin{equation}
S_{\rm bh} \propto \frac{A}{\ell_P^2}\,,
\end{equation}
where $\ell_P$ is the Planck length. Moreover, he formulated the generalized second law of thermodynamics \cite{Bekenstein:1974ax}
\begin{equation}
 \Delta S_{\rm gen}\equiv  \Delta S_{\rm bh} + \Delta S_{\rm ext} \geq 0\,,
\end{equation}
where $S_{\rm gen}$ is the generalized entropy and $S_{\rm ext}$ is the entropy of the exterior universe.

However, the assignment of entropy to the black hole is not sufficient to resolve the Geroch paradox. As we saw above, the energy delivered to the black hole can be made arbitrarily small, so that the increase in black hole entropy cannot compensate for the entropy lost from the exterior. Bekenstein's resolution \cite{Bekenstein:1973ur,Bekenstein:1974ax} of this problem was to argue that the Geroch process cannot be carried out arbitrarily close to the horizon, implying a lower nonzero-bound  on the energy delivered to the black hole. He later refined his resolution \cite{Bekenstein:1980jp}: requiring that the resulting increase in black hole entropy be at least as large as the entropy $S$ of the system   lead him to a bound on the entropy  of a system with a given energy and size  
\begin{equation}\label{bekensteinbound}
S \leq 2\pi m R\,,
\end{equation}
This bound was   reformulated and proven in quantum field theory by Casini \cite{Casini:2008cr} as a consequence of the positivity of relative entropy for states reduced to the Rindler wedge, with $S$ interpreted as vacuum-subtracted entanglement entropy. However, Casini's bound is not needed to derive the generalized second law, which can instead be obtained from the monotonicity of relative entropy, as shown by Wall \cite{Wall:2010cj,Wall:2011hj}.

A different resolution of the Geroch paradox was provided by Unruh and Wald~\cite{Unruh:1982ic}, who argued that no independent entropy bound, such as \eqref{bekensteinbound}, is required to save the generalized second law once quantum effects are properly taken into account. For a static black hole in thermal equilibrium, a static observer outside the black hole sees a thermal bath of radiation at the locally measured Tolman temperature
$ 
T_{\rm loc}= T_{\rm H}/ \chi,
$ 
where \(T_{\rm H}\) is the Hawking temperature. Since the redshift factor $\chi$ is not constant, the thermal atmosphere has a nonzero gradient in the locally measured temperature, and hence in the pressure.

Consequently, when a box is slowly lowered toward the black hole, it experiences an effective buoyancy force arising from the surrounding thermal radiation, with a larger pressure exerted on the face of the box closer to the horizon. As a result, less work is delivered to infinity during the lowering process, and hence more energy is delivered to the black hole than in the classical analysis. The energy delivered is minimized not at the horizon, but at a floating point where the box is in equilibrium with the surrounding atmosphere. The location of the floating point is determined by   Archimedes' principle that  the energy of the box equals that of the displaced thermal radiation. When the contents of the box are released into the black hole at the floating point, the minimal increase in black hole entropy equals the entropy of the displaced thermal radiation. Since unconstrained thermal radiation maximizes entropy at fixed energy and volume, the entropy contained in the box cannot exceed that of the displaced radiation. Hence, the generalized entropy cannot decrease.

In this way, the apparent violation of the generalized second law is resolved by the thermal atmosphere surrounding the black hole. Unruh and Wald~\cite{Unruh:1983ir} extended their analysis to more general configurations, including thick boxes, and showed that the generalized second law holds in this process without the need to assume an independent entropy bound. These conclusions were subsequently challenged by Bekenstein in a series of critiques~\cite{Bekenstein:1983iq,Bekenstein:1993dz,Bekenstein:1999bh}. Pelath and Wald~\cite{Pelath:1999xt}  later demonstrated that, under the same assumptions about unconstrained thermal matter and the black hole atmosphere used in Bekenstein's analysis, the generalized second law can be established without introducing a separate entropy bound. See also the  reviews~\cite{SCIAMA1976385,Landsberg1992,Wald:1999vt,Anderson:1999gsl} on this topic, and we refer to   \cite{PhysRevD.43.340,Richterek,Curiel:2014zua,Bravetti:2015xsp,Prunkl:2019wdw} for alternative approaches to black hole heat engines.\\

\noindent \textbf{Holographic heat engines.} These early developments establish that black holes behave as consistent thermodynamic systems and can participate in heat-engine-like processes. A modern realization of this idea was developed by Johnson, who proposed to treat black holes themselves as working substances in heat engines within extended black hole thermodynamics. This framework is based on an extended version of the first law for Anti-de Sitter black holes where the cosmological constant $\Lambda$ is allowed to vary \cite{Kastor:2009wy,Dolan:2010ha,Dolan:2011xt,Cvetic:2010jb,Kubiznak:2014zwa}. The cosmological constant is then interpreted as a bulk pressure, $P_{\Lambda} = - \Lambda / (8\pi G)$, and its conjugate quantity in the extended first law defines the thermodynamic volume $V_\Lambda$. In this approach the black hole mass is identified with the thermodynamic enthalpy: $M= H= E+ P_\Lambda V_\Lambda.$ These definitions of pressure and volume   allow one to consider thermodynamic cycles in the $(P_{\Lambda},V_\Lambda)$-plane that extract work from AdS black holes used as the working substance.

In his original work \cite{Johnson:2014yja}, Johnson considered heat engines based on charged AdS black holes. For static black holes, a key simplification occurs: both the entropy and thermodynamic volume depend only on the horizon radius, so they are not independent variables. As a result, isochores coincide with adiabats. Consequently, the Carnot cycle (consisting of two isotherms and two adiabats) and Stirling cycle (consisting of two isotherms and two isochores) are identical in this framework and both attain the Carnot efficiency. A cycle that is not identical to the Carnot cycle is the rectangular $P_\Lambda V_\Lambda$-cycle  composed of isochores and isobars, whose efficiency Johnson evaluated in a high pressure and temperature approximation. Later  he \cite{Johnson:2016pfa} derived     an exact efficiency  formula for such rectangular cycles  in terms of differences of the black hole mass   evaluated at the corners of the cycle.

Subsequent work \cite{Chakraborty:2016ssb} placed these constructions on a more systematic footing by introducing a benchmarking scheme for black hole heat engines. In particular, a circular cycle in the $(P_{\Lambda},V_\Lambda)$-plane was proposed as a universal testbed, together with a discretization procedure that allows one to compute efficiencies for arbitrary cycles and compare different black hole working substances within a common framework.
More recently, \cite{Johnson:2019olt} reformulated these setups in the language of quantum heat engines, treating black holes explicitly as quantum working substances and analyzing standard engine cycles including the Brayton, Otto, and Diesel cycles, without modifying the   thermodynamic framework. Other extensions include higher-curvature corrections \cite{Johnson:2015ekr} and rotating black holes \cite{Hennigar:2017apu}, among many others; see \cite{Mann:2025xrb} for a comprehensive review of extended black hole thermodynamics, including heat engines. 

As noted by Johnson in his original work \cite{Johnson:2014yja}, his construction of heat engines can be interpreted within the framework of the Anti de Sitter/Conformal Field Theory (AdS/CFT) correspondence \cite{Maldacena:1997re,Gubser:1998bc,Witten:1998qj}. In particular, the bulk pressure is tied to the cosmological constant, so varying it corresponds to varying   the number of field degrees of freedom of the CFT (or the central charge). In this sense, a thermodynamic cycle in the bulk may be viewed as a trajectory through a family of dual field theories, similar to a renormalization group flow. This, however, raises a conceptual issue: the bulk pressure is not the thermodynamic pressure of the dual system, but instead parametrizes changes in the theory itself. As a result, the working substance does not evolve within a fixed quantum system, but rather moves in the space of theories, and therefore such cycles do not straightforwardly correspond to standard thermodynamic processes in the dual field theory \cite{Mancilla:2024spp,Borsboom:2026ash}.  

In contrast, in our recent construction of holographic heat engines \cite{Lilani:2025gzt}, the working substance is a thermal equilibrium state of a conformal field theory, and pressure and volume are defined in the standard thermodynamic sense intrinsic to the boundary theory. The heat engines are therefore formulated directly in the CFT, at fixed central charge, and their cycles correspond to ordinary quasi-static thermodynamic processes within a fixed quantum theory. In this setting, the equation of state is determined by scale invariance, leading to the conformal relation $E=(D-1)PV$. As a result, the efficiencies of a broad class of idealized, reversible cycles, including the Brayton, Otto, Diesel, and rectangular cycles, are universally fixed by the equation of state and are therefore independent of the field content of the CFT.

An important   difference from earlier constructions is that entropy and volume are independent thermodynamic variables. Consequently, the degeneracy between isochores and adiabats that arises in extended black hole thermodynamics is absent, and the Carnot and Stirling cycles are genuinely distinct. In particular, the efficiency of a Stirling cycle is non-trivial for a CFT: in the absence of regeneration, heat is exchanged along all four branches, and its efficiency cannot be determined solely from the equation of state in terms of the characteristic parameters temperature and volume. Instead, it depends on additional thermodynamic data, such as the functions $S(T,V)$ and $P(T,V)$. Universal behavior is recovered only in the high-temperature or large-volume regime. At finite temperature and volume   the Stirling efficiency depends on the subextensive coefficients appearing in the   large-temperature or large-volume expansion of the canonical free energy \cite{Lilani:2025gzt,Kutasov:2000td}.

Although our construction applies to thermal states in general CFTs, a particularly important regime arises when the state admits a dual description in terms of an asymptotically AdS black hole. In this case, the heat engine acquires a direct holographic interpretation, providing a conceptually cleaner realization of holographic heat engines. This regime corresponds to strongly coupled, large-$N$ CFTs, for which a precise holographic dictionary relates bulk and boundary thermodynamic quantities \cite{Gubser:1998bc,Witten:1998qj,Witten:1998zw,Visser:2021eqk}. Exploiting this dictionary for AdS-Schwarzschild black holes, we obtained a closed-form expression for the non-regenerative Stirling efficiency at finite temperature and volume. The resulting efficiency is lower than the planar (infinite-volume) limit. Moreover, comparing with weakly coupled CFTs, we found that the non-regenerative Stirling efficiency is likewise reduced in the holographic (strongly coupled) regime.\\

\noindent \textbf{Summary of results.} In the present work, we extend our previous analysis of CFT and holographic heat engines by incorporating regeneration into the Stirling cycle, which is in fact how Robert Stirling originally introduced the engine in 1816,\footnote{We thank Bill Unruh for pointing this out to us, and for pressing us to include regeneration.} whereas in our earlier work \cite{Lilani:2025gzt} we   considered only a non-regenerative variant. Regeneration is an internal heat recycling mechanism in which heat released during one part of the cycle is stored and later returned to the working substance, thereby reducing the net heat exchange with external reservoirs and increasing the efficiency. \textcolor{black}{By separating the heat exchanged with external reservoirs from the heat
internally recycled by the regenerator, we derive energy balance expressions
for regenerative Stirling cycles and state explicitly the additional condition
required for reversible regeneration: the heat released during cooling must
equal the heat absorbed during heating at each temperature
(see section~\ref{sec:2.1}).}

A central role in our analysis is played by the intrinsic mismatch between the
heat exchanged along the two isochoric branches of the Stirling cycle,
\begin{equation} \label{heatmismatch1}
Q_{\rm mis}
\equiv
\big|Q^{\rm out}_{2\to 3}-Q^{\rm in}_{4\to 1}\big|\,.
\end{equation}
\textcolor{black}{This quantity measures the net energetic difference between the
heat released during isochoric cooling and the heat required during isochoric
heating. It therefore determines the net residual heat that must be exchanged
with the external reservoirs in order to restore the energy balance of the
cycle. However, $Q_{\rm mis}=0$ alone does not guarantee reversible
regeneration, since equality of the total heats does not ensure that heat is
returned at the same temperatures at which it was stored.}

\textcolor{black}{We identify a simple and general sufficient condition for
reversible regeneration, namely that, at fixed conserved charges, the
fixed-volume heat capacity is independent of volume,}
\begin{equation} \label{conditionno1}
C_V(T,V,\mathcal Q_i)=C_V(T,\mathcal Q_i)\,,
\end{equation}
where the conserved charges $\mathcal Q_i$ are held fixed.
\textcolor{black}{Condition~\eqref{conditionno1} implies that the heat released
along the isochore $2\to3$ equals the heat reabsorbed along $4\to1$ at each
temperature. Hence $Q_{\rm mis}=0$, and the regenerator returns to its initial
thermodynamic state with no net change in energy or entropy. We then show that
the Stirling efficiency equals the Carnot efficiency (see the theorem in
section~\ref{sec:stirlingtheorem}):}
\begin{equation}
\eta_{\rm Stirling}^{\rm regen}
=
\eta_{\rm Carnot}
=
1-\frac{T_{\rm c}}{T_{\rm h}}\,.
\end{equation}
\textcolor{black}{This result can be understood directly from Carnot's theorem,
according to which all reversible engines that exchange heat only isothermally
with a hot and a cold reservoir, at temperatures $T_{\rm h}$ and $T_{\rm c}$,
respectively, attain the Carnot efficiency. Under
condition~\eqref{conditionno1}, the heat released by the working substance to
the regenerator at each temperature along $2\to3$ exactly matches the heat
reabsorbed from it at the same temperature along $4\to1$. The internal heat
exchange can therefore be carried out reversibly, and the only heat exchanged
with the two external reservoirs occurs along the isotherms. We emphasize that
the weaker condition $Q_{\rm mis}=0$ alone does not suffice: it balances only
the total energy given to and received from the regenerator and does not ensure
zero net entropy change.}

It is well known that the Carnot efficiency is attained for a regenerative
Stirling cycle if the working substance consists of an ideal gas
(in section~\ref{sec:idealgas} we review this textbook result to set the
stage), and we show that the same holds for the Van der Waals fluid
(see section~\ref{sec:waals}).

\begin{figure}[tbp]
    \centering
    \begin{subfigure}[t]{0.45\textwidth}
        \centering
        \includegraphics[width=\textwidth]{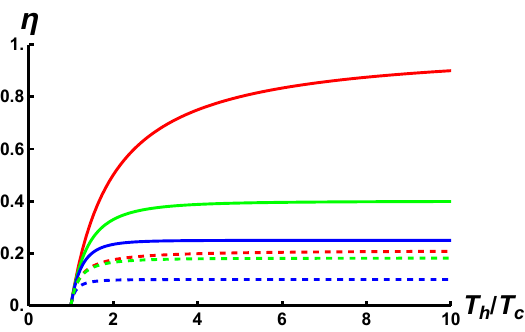} 
        \caption{Stirling efficiency $\eta$ vs. temperature ratio $T_{\rm h}/T_{\rm c}$ at fixed compression ratio, for an ideal gas (red), Bose-Einstein condensate (green), and CFT on a plane (blue), shown both in the presence (solid) and absence (dashed) of regeneration. $V_2/V_1 = 1.5$ and $D=4$ for this plot. The red solid line corresponds to the Carnot efficiency.}
        \label{fig:sub1}
    \end{subfigure}
    \hfill
    \begin{subfigure}[t]{0.45\textwidth}
        \centering
        \includegraphics[width=\textwidth]{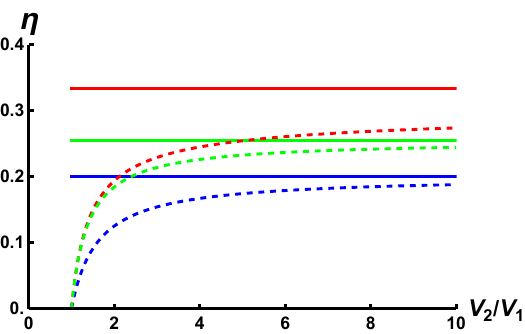} 
        \caption{Stirling efficiency $\eta$ vs. compression ratio $V_2/V_1$ at fixed temperature ratio, for an ideal gas (red), Bose-Einstein condensate (green), and CFT on a plane (blue), shown both in the presence (solid) and absence (dashed) of regeneration. $T_{\rm h}/T_{\rm c} = 1.5$ and $D=4$ for this plot. }
        \label{fig:sub2}
    \end{subfigure}

    \caption{ Comparison of Stirling efficiencies in the presence and absence of regeneration.}
    \label{fig:overallno1}
\end{figure}

 \textcolor{black}{For generic systems condition~\eqref{conditionno1} is not satisfied. Volume dependence of $C_V$ alone does not logically imply $Q_{\rm mis}\neq0$, because the integrated heat exchanges could cancel accidentally; however, our explicit calculations for quantum ideal gases away from the classical limit and for the neutral CFT working substances show that the isochoric heat exchanges do not match and that the regenerative efficiency remains below the Carnot value.}

We compute the Stirling efficiency explicitly for both regenerative and non-regenerative cycles, for quantum ideal gases (see section \ref{sec:quantumgases}), generic thermal CFTs (see section \ref{sec:thermalcfts}) and holographic CFTs (see section \ref{sec:holographiccfts}).  For quantum gases, we analyze both Bose and Fermi ideal  gases and show that, in the regenerative case, the deviation from Carnot efficiency is controlled by the fugacity~$z$, with the classical limit $z \to 0$ recovering Carnot efficiency. In the Bose case, this includes the regime of Bose-Einstein condensation (BEC), where the regenerative Stirling efficiency takes the simple form
\begin{equation}
\eta_{\rm regen}^{\rm BEC}
= \frac{1 - (T_{\rm c}/T_{\rm h})^{1 + d/2}}{1 + d/2} \, ,
\end{equation}
where $d$ is the number of spatial dimensions.

\textcolor{black}{For the neutral CFT working substances analyzed in section~\ref{sec:thermalcfts}, the explicit isochoric heat exchanges do not match and the Stirling cycle does not attain the Carnot efficiency. Their volume-dependent fixed-volume heat capacity explains why the sufficient condition~\eqref{conditionno1} is unavailable, but the conclusion follows from the explicit calculation rather than from volume dependence alone.} The Stirling efficiency simplifies in the high-temperature or planar limit, for instance in the regenerative case we obtain  
\begin{equation}
\eta^{\rm CFT}_{\rm regen} = \frac{1}{D}\left(1 - \left(\frac{T_{\rm c}}{T_{\rm h}}\right)^D\right)\,,
\end{equation}
where $D$ is the number of CFT spacetime dimensions. At finite temperature and volume we can compare the CFT Stirling efficiency at zero and infinite 't Hooft coupling. For example,  for $\mathcal{N} = 4$ SYM theory we find:
\begin{equation}\eta^{\text{regen}}_{\lambda =0} > \eta^{\text{regen}}_{\lambda \rightarrow \infty} > \eta^{\text{non-regen}}_{\lambda =0} > \eta^{\text{non-regen}}_{\lambda \rightarrow \infty}\,.\end{equation} Thus, we see that regeneration and weak coupling enhance the   CFT Stirling efficiency.

For holographic CFTs, we consider two bulk geometries, namely
AdS-Schwarzschild and AdS-Reissner-Nordstr\"om black holes. In the fixed
electric potential ensemble, both efficiencies approach the Carnot value
at large potential, with corrections of order
$\mathcal O(\tilde\Phi^{-1})$. The coefficient of the leading correction
is smaller in the presence of regeneration. This occurs outside the
assumptions of the theorem in section~\ref{sec:stirlingtheorem}, because
the electric charge varies along the fixed-potential cycle.

Finally, we compare the Stirling efficiencies, with and without regeneration, for the classical ideal gas, Bose-Einstein condensate, and CFT on the plane. As shown in Figure~\ref{fig:overallno1}, regeneration restores Carnot efficiency for the ideal gas and enhances the efficiency for BEC and CFT systems. Furthermore, the CFT efficiency is systematically lower than that of the BEC, which in turn is lower than that of the ideal gas, both as a function of the temperature ratio and the compression ratio.

\section{Stirling engine with and without regeneration}

\subsection{Reversible Stirling cycle and its  efficiency} 
\label{sec:2.1}
\noindent \textbf{Reversible heat engines.} A heat engine consists of four main components: a working substance (or system), a heat source, a heat sink and a work output device.   Moreover, an engine potentially also has a regenerator, a heat reservoir that  temporarily    stores and later expels heat to improve efficiency, which plays an important role in this article.   The essence of a heat engine is that the working substance   converts heat into work and operates in  a thermodynamic cycle. During the cycle, heat $(Q_{\text{in}})$ is supplied from the heat source to the working substance, part of which is converted into work $(W)$ done on the work output device, and the remainder heat $(Q_{\text{out}})$ is dumped into a heat sink. These three quantities are constrained by the first law of thermodynamics
\begin{equation}
Q_{\text{in}} = W + Q_{\text{out}}\,,
\end{equation} 
where we take $Q_{\text{in}} $ and $Q_{\text{out}}$ (with subscripts) to be positive by convention, and $W$ is positive if the working substance performs work and negative if work is done on the working substance. We consider heat sources and sinks consisting of one or more thermal reservoirs, which have a very large heat capacity so  that heat exchanges do not alter their temperature. 

The efficiency of a heat engine is defined as the ratio of the work performed by the system and the heat put into the system
\begin{equation} \label{efficiency}
    \eta = \frac{W}{Q_{\text{in}}} = 1- \frac{Q_{\text{out}}}{Q_{\text{in}}}\,,
\end{equation}
where the first law was used in the second equality.  \textcolor{black}{We assume that the cycle consists of reversible thermodynamic processes.
In particular, the processes are quasi-static, so that the working substance
remains in equilibrium throughout, and have zero entropy production.
For a quasi-static process involving only pressure-volume work and generalized
work terms associated with the quantities $\mathcal Q_i$, the first law takes
the differential form
\begin{equation}
    dE=\delta Q-P\,dV+\sum_i \Psi_i\,d\mathcal Q_i\,,
\end{equation}
where $\delta Q$ (without subscript in/out) is positive when it is added to the system and negative when it leaves, 
 and
$\Psi_i$ is the generalized potential conjugate to $\mathcal Q_i$. Throughout
most of this paper, except for
section~\ref{subsec:stirling_equal_to_carnot}, the quantities $\mathcal Q_i$
are held fixed. For a reversible process, the Clausius relation is
\begin{equation}
    \delta Q_{\rm rev}=T\,dS\,,
\end{equation}
so that $Q_{\rm rev}=\int T\,dS$ for a finite path, which reduces to
$Q_{\rm rev}=T\Delta S$ for an isothermal process.
In this article, we study the Stirling engine for different kinds of working
substances.}\\

\noindent \textbf{Stirling cycle.} We first define a Stirling cycle in abstract terms, and then describe   more concrete realizations  of a Stirling engine (see \cite{Walker1980,UrieliBerchowitz1984,FinkelsteinOrgan2001,Senft2007} for standard references). A Stirling cycle is a closed thermodynamic cycle that consists of four paths. We label the vertices of the four paths by $i=1,2,3,4$. The thermodynamic processes along the four paths of a Stirling cycle are: $1 \to 2 $ and $3 \to 4$ are isotherms, i.e. they have the same temperature  $T_1 = T_2$ and $T_3 = T_4$, and paths $2 \to 3$ and $4 \to 1$ are isochores, i.e. they have constant volume $V_2=V_3$ and $V_4 = V_1$.  Unlike internal heat engines, such as the gasoline and diesel engine,  where heat is generated inside the working fluid, the Stirling engine is an external heat engine, like the steam engine: 
the heat that drives the cycle is supplied from outside the engine's working fluid and crosses the system boundary during heat transfer. The working substance of a  Stirling engine typically consists of air or another gas.

We define two types of Stirling engines: with and without a regenerator. 
Historically, the 1816 patent of Robert Stirling~\cite{Stirling1816} introduced 
a closed-cycle air engine,  operating on what became known as the 
Stirling cycle, together with the ``economiser,'' a heat-exchange device 
designed to improve efficiency by storing and reusing heat within the cycle \cite{Walker1980,FinkelsteinOrgan2001}. 
This component, now known as the regenerator (due to Ericsson), was subsequently refined and 
incorporated into practical engine designs, including those developed in 
collaboration with his brother James Stirling. For convenience, albeit with 
some historical imprecision, we refer to engines lacking a regenerator as 
nonregenerative Stirling engines. Regeneration refers to the internal thermal 
process in which heat is temporarily stored during one part of the cycle and 
returned in another, without exchange with an external reservoir. As a result 
of this internal recycling of heat, the required external heat input is reduced 
for a given cycle. In a Stirling engine, heat exchange with the regenerator 
occurs during the isochoric stages of the cycle.

There are three common mechanical realizations of the Stirling engine: the alpha, beta, and gamma configurations \cite{UrieliBerchowitz1984}.  In the \emph{alpha-type} engine, the working gas is contained in two separate cylinders connected by a passage that may include a regenerator. One cylinder is maintained at the hot temperature $T_{\rm h}$ and the other at the cold temperature $T_{\rm c}$, each containing a piston coupled to a crankshaft. In the \emph{beta-type} engine, a power piston and a displacer are located in the same cylinder. The displacer does not produce work. Rather, it moves the working gas between the hot and cold regions of the cylinder, typically through a regenerator, while the power piston converts the resulting pressure changes into mechanical work. In the \emph{gamma-type} engine, the displacer and the power piston are placed in separate cylinders that are connected by a passage containing the working gas. The displacer cylinder contains the hot and cold regions and moves the gas between them, while the power piston cylinder extracts work from the pressure changes in the gas.

\begin{figure}[t!]
    \centering
    \includegraphics[width=1.1\textwidth]{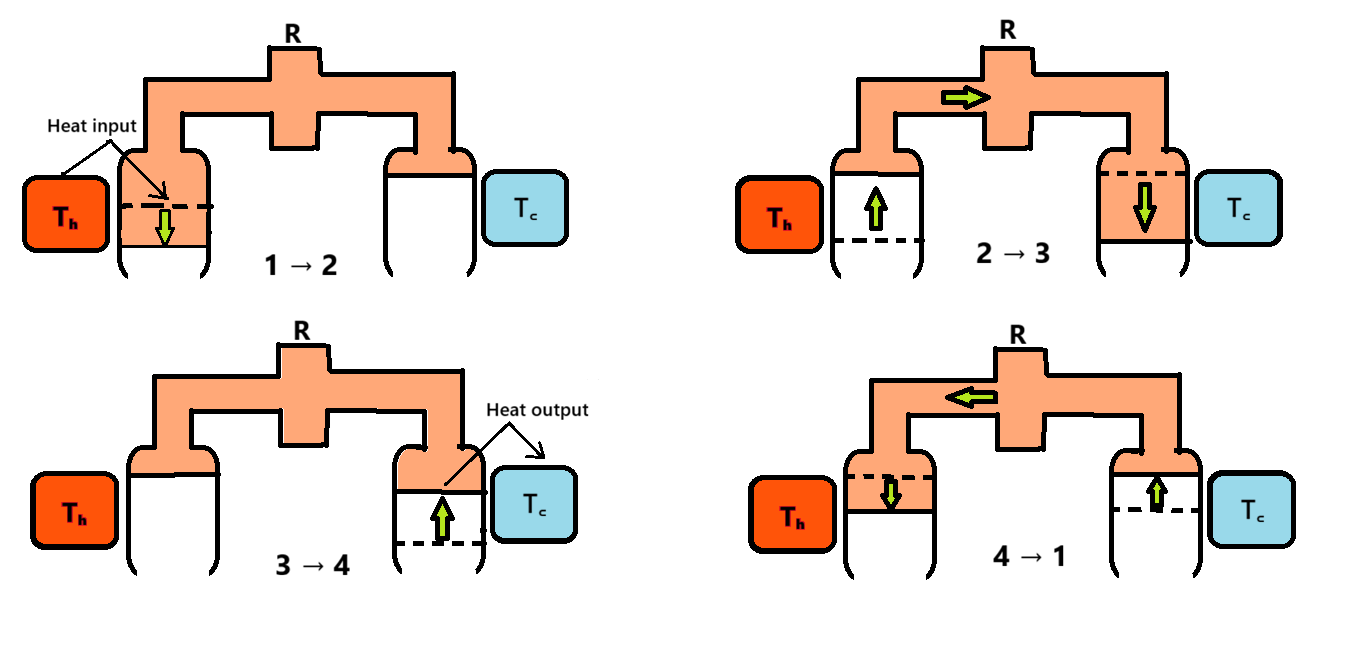}
    \caption{Operation of an  alpha-type   Stirling engine with a regenerator.}
    \label{fig:my_figurenew}
\end{figure}

For concreteness, we describe the four steps of an idealized alpha-type Stirling cycle (see also Figure \ref{fig:my_figurenew}): 

 \begin{itemize}

\item[(i)] $1 \to 2$: \emph{Isothermal expansion} (at the hot temperature $T_{\text{h}}$). The working substance in the hot cylinder absorbs heat from the heat reservoir at temperature $T_\text{h}$. The gas expands isothermally from volume $V_1 $ to $V_2$, pushing the hot piston outward and thereby performing work. The piston in the cold cylinder remains at rest. 

\item[(ii)] $2 \to 3$: \emph{Isochoric cooling} (heat rejection to regenerator or cold reservoir). The hot piston moves in while the cold piston moves out, transferring the gas between the cylinders while keeping the total volume constant. In the presence of regeneration, the gas flows past the regenerator, giving up heat and cooling to $T_{\text{c}}$.  In the absence of regeneration, the gas instead releases heat to the cold reservoir and cools to $T_{\text{c}}$.  Maintaining reversibility in this case would require a continuum of cold reservoirs so that the temperature difference remains infinitesimal throughout the process.

\item[(iii)] $3 \to 4$: \emph{Isothermal compression} (at the cold temperature $T_{\text{c}}$). The cold piston moves inward, compressing the gas isothermally from volume $V_2$ back to its original volume~$V_1$. 
The gas releases heat to the cold reservoir while remaining at temperature~$T_{\text{c}}$. 
During this process work is done on the gas, while the hot piston remains fixed.

\item[(iv)] $4 \to 1$: \emph{Isochoric heating} (heat absorption  from regenerator or hot reservoir). 
The cold piston moves the rest of the way inward while the hot piston moves outward, transferring the gas back to the hot cylinder while keeping the total volume constant.  In the presence of regeneration, the gas passes the regenerator and reabsorbs the stored heat until its temperature reaches $T_{\text{h}}$.  In the absence of regeneration, the gas instead absorbs heat directly from the hot reservoir, which would require a continuum of hot reservoirs for the process to remain reversible.

\end{itemize}

\noindent For later reference,  we summarize   the relations between the different temperatures and   between the volumes 
\begin{equation}
T_1=T_2=T_{\rm h}\,,\qquad T_3=T_4=T_{\rm c}\,,\qquad V_2=V_3\,,\qquad V_1=V_4\,,
 \end{equation}
 and $T_{\rm h}>T_{\rm c}$ and $V_2>V_1$.

  To recapitulate, in the absence of a regenerator, heat is supplied from external reservoirs during processes $1\to2$ and $4\to1$, and is expelled during $2\to3$ and $3\to4$. \textcolor{black}{With pointwise matched regeneration, the isochoric exchange is entirely internal: external heat enters on $1\to2$ and leaves on $3\to4$. More generally, a local isochoric surplus or deficit must also be exchanged with auxiliary reservoirs at the corresponding temperatures, as described below.}  \\

\noindent \textbf{Stirling efficiency.}
The efficiency of the Stirling cycle can be computed in two ways, using the two expressions in \eqref{efficiency}. 
We denote by \(Q^{i\to j}_{\text{in}}\) and \(Q^{i\to j}_{\text{out}}\) the heat absorbed or released by the working substance along the path \(i\to j\), and by \(W_{i\to j}\) the work done by the system  along that path.  Along isotherms the heat transfer is given by 
  the Clausius relation \(Q = T\Delta S\), while  along isochores with the volume and conserved charges fixed, the heat transfer equals the change in internal energy, \(Q=\Delta E\).

In the absence of regeneration, heat enters the system along the paths \(1\to2\) and \(4\to1\), and leaves the system along the paths \(2\to3\) and \(3\to4\).  
The efficiency can then be written as
\begin{equation} \label{nonregen1}
    \eta_{\text{non-regen}}
    = \frac{W_{1\to2}+W_{3\to4}}{Q_{\text{in}}^{1\to2}+Q_{\text{in}}^{4\to1}}
    = \frac{\int_{V_1}^{V_2}[P(V,T_{\text{h}})-P(V,T_{\text{c}})]\,dV}
    {T_{\text{h}}\Delta S_{1\to2}+\Delta E_{4\to1}}\,,
\end{equation}
where the numerator represents the net work done during the two isothermal processes. 
Equivalently, using \(\eta = 1-Q_{\text{out}}/Q_{\text{in}}\), we may write
\begin{equation} \label{nonregen2}
    \eta_{\text{non-regen}}
    = 1-\frac{Q_{\text{out}}^{2\to3}+Q_{\text{out}}^{3\to4}}
    {Q_{\text{in}}^{1\to2}+Q_{\text{in}}^{4\to1}}
    = 1-\frac{\Delta E_{2\to3}+T_{\text{c}}\Delta S_{3\to4}}
    {T_{\text{h}}\Delta S_{1\to2}+\Delta E_{4\to1}} \,.
\end{equation}
 The quantities $\Delta S$ and $\Delta E$ appearing on the right-hand side are positive magnitudes, consistently with the conventions for $Q_{\rm in}$ and $Q_{\rm out}$.  

In the presence of regeneration, the heat expelled during isochoric cooling is stored internally in a regenerator and later returned to the working substance during isochoric heating. This reduces the net heat exchange with external reservoirs along the isochoric steps, although it does not reduce it to zero generically (see e.g. \cite{Huang2014,Gupt:2021iir,Altintas:2026jxc}). The positive quantity
\begin{equation} \label{defregen}
Q_{\rm mis}
\equiv
\left|Q^{2\to3}_{\rm out}-Q^{4\to1}_{\rm in}\right|
\end{equation}
\textcolor{black}{measures the net energetic difference between the two
isochoric branches. For regeneration to be reversible, the heat released
during cooling at each temperature must equal the heat required during heating
at that same temperature. At fixed conserved charges, define the local
difference
\begin{equation}
dQ_{\rm loc}(T)
=
\big[
C_V(T,V_2,\mathcal Q_i)
-
C_V(T,V_1,\mathcal Q_i)
\big]\,dT,
\qquad
Q_{\rm mis}
=
\left|
\int_{T_{\rm c}}^{T_{\rm h}} dQ_{\rm loc}(T)
\right|.
\end{equation}
If $dQ_{\rm loc}(T)$ has the same sign throughout the temperature interval,
then $Q_{\rm mis}$ is the total amount of heat that cannot be recycled
internally, and the efficiency formulas below apply. If the sign changes,
surplus heat at some temperatures can cancel a heat deficit at other
temperatures in the integral defining $Q_{\rm mis}$. In that case,
$Q_{\rm mis}$ alone does not determine the total heat that must be exchanged
externally. The formulas below should then be interpreted only as net
energy-balance expressions, and a reversible treatment must keep track of the
surplus and deficit separately at each temperature. Any unmatched heat must be
exchanged with external reservoirs at the corresponding temperatures. In the applications below, we use the efficiency formulas directly only
when $dQ_{\rm loc}(T)$ has a fixed sign throughout the relevant temperature
interval.}


In addition, realistic regenerators are not $100\%$ effective: only a fraction
of the heat available during isochoric cooling is recovered and returned during
isochoric heating. This is typically characterized by a regenerator
effectiveness $\varepsilon_R<1$, leading to additional heat exchange with the
external reservoirs and a corresponding reduction in efficiency. In this work,
we assume an ideal regenerator with $\varepsilon_R=1$ and focus instead on the
intrinsic mismatch between the heat released during isochoric cooling and the
heat required during isochoric heating. This intrinsic mismatch should be
distinguished from losses caused by finite regenerator effectiveness.

\textcolor{black}{In particular, if $Q^{2\to3}_{\rm out}>Q^{4\to1}_{\rm in}$ and the local mismatch has the corresponding fixed sign, the regenerator stores excess heat during cooling. The surplus is rejected externally at the temperatures at which it occurs so that the regenerator returns to its initial state.} In this case the efficiency becomes
\begin{equation} \label{regen1}
    \eta_{\text{regen}}
    = \frac{W_{1\to2}+W_{3\to4}}{Q_{\text{in}}^{1\to2}}
    = \frac{\int_{V_1}^{V_2}[P(V,T_{\text{h}})-P(V,T_{\text{c}})]\,dV}
    {T_{\text{h}}\Delta S_{1\to2}}\,,
\end{equation}
or equivalently
\begin{equation} \label{regen2}
    \eta_{\text{regen}}
    = 1-\frac{Q_{\text{out}}^{3\to4}+Q_{\text{mis}}}{Q_{\text{in}}^{1\to2}}
    = 1-\frac{T_{\text{c}}\Delta S_{3\to4}+Q_{\text{mis}}}
    {T_{\text{h}}\Delta S_{1\to2}}\,,
\end{equation}
where $  Q_{\text{mis}}$ 
denotes the excess heat  temporarily stored in the regenerator, but later expelled.

\textcolor{black}{In the opposite fixed-sign case, $Q^{2\to3}_{\rm out}<Q^{4\to1}_{\rm in}$, the stored heat is insufficient and an additional amount $Q_{\rm mis}$ must be supplied externally at the corresponding temperatures.} 
The efficiency then becomes
\begin{equation} \label{regen3}
    \eta_{\text{regen}}
    = \frac{W_{1\to2}+W_{3\to4}}{Q_{\text{in}}^{1\to2}+Q_{\text{mis}}}
    = \frac{\int_{V_1}^{V_2}[P(V,T_{\text{h}})-P(V,T_{\text{c}})]\,dV}
    {T_{\text{h}}\Delta S_{1\to2}+Q_{\text{mis}}}\,,
\end{equation}
or equivalently
\begin{equation} \label{regen4}
    \eta_{\text{regen}}
    = 1-\frac{Q_{\text{out}}^{3\to4}}{Q_{\text{in}}^{1\to2}+Q_{\text{mis}}}
    = 1-\frac{T_{\text{c}}\Delta S_{3\to4}}
    {T_{\text{h}}\Delta S_{1\to2}+Q_{\text{mis}}}\,.
\end{equation}
\textcolor{black}{We emphasize that, in these two fixed-sign cases, an external exchange of magnitude $Q_{\rm mis}$ is needed for both the working substance and the regenerator to return to their initial states.} For instance, for  CFT engines  $Q_{\text{mis}}$ is generically nonzero. However, there exist working substances for which $Q_{\text{mis}}=0$. Classical ideal gases and Van der Waals fluids 
are examples of such working substances.   We will show this explicitly   in the next sections.

\subsection{Theorem on the efficiency of the Stirling engine}
\label{sec:stirlingtheorem}

\textcolor{black}{As noted above, the quantity $Q_{\rm mis}$ need not vanish for a general working substance. We now identify a criterion under which the local isochoric mismatch vanishes pointwise, and hence $Q_{\rm mis}=0$.}

\begin{prop}
\textcolor{black}{For an ideal regenerative Stirling engine with conserved charges \(\mathcal{Q}_i\) held fixed throughout the cycle, the local isochoric heat mismatch vanishes at every temperature (and hence \(Q_{\rm mis}=0\)) if the constant-volume heat capacity is independent of volume, i.e.} 
\begin{equation} \label{criterionheatcapacity}
    C_V(T,V,\mathcal{Q}_i) = C_V(T,\mathcal{Q}_i)\,.
\end{equation}
\end{prop}


\begin{proof}
Along an isochore with the conserved charges $\mathcal Q_i$ held fixed, no
mechanical or generalized work is performed. The reversible heat exchange is
therefore determined by the constant-volume heat capacity,
\begin{equation}
C_V(T,V,\mathcal Q_i)
\equiv
T\left(\frac{\partial S}{\partial T}\right)_{V,\mathcal Q_i}
=
\left(\frac{\partial E}{\partial T}\right)_{V,\mathcal Q_i}.
\end{equation}
It follows that the amounts of heat released during isochoric cooling and
absorbed during isochoric heating are
\begin{equation}
Q_{\rm out}^{2\to3}
=
\int_{T_{\rm c}}^{T_{\rm h}}
C_V(T,V_2,\mathcal Q_i)\,dT,
\qquad
Q_{\rm in}^{4\to1}
=
\int_{T_{\rm c}}^{T_{\rm h}}
C_V(T,V_1,\mathcal Q_i)\,dT.
\end{equation}
Here both integrals are written from $T_{\rm c}$ to $T_{\rm h}$ because
$Q_{\rm out}^{2\to3}$ and $Q_{\rm in}^{4\to1}$ denote positive amounts of
heat.

\textcolor{black}{The difference between the heat exchanged by the two isochores in a small
temperature interval is
\begin{equation}
dQ_{\rm loc}(T)
=
\left[
C_V(T,V_2,\mathcal Q_i)
-
C_V(T,V_1,\mathcal Q_i)
\right]dT.
\end{equation}
If the heat capacity is independent of volume,
 Eq. \eqref{criterionheatcapacity},
then
$
dQ_{\rm loc}(T)=0
$
at every temperature between $T_{\rm c}$ and $T_{\rm h}$. Thus, the heat
released during cooling equals the heat required during heating not only
after integration, but at each temperature along the isochores.
Integrating the local difference gives
\begin{equation}
Q_{\rm out}^{2\to3}-Q_{\rm in}^{4\to1}
=
\int_{T_{\rm c}}^{T_{\rm h}}dQ_{\rm loc}(T)
=0.
\end{equation}
By definition \eqref{defregen}, this implies \(Q_{\text{mis}}=0\).}
\end{proof}

\noindent
We now show that, under the same condition, the efficiency of the regenerative Stirling engine coincides with the Carnot efficiency,
\begin{equation}
\eta_{\rm regen} = \eta_{\rm Carnot}
\qquad \text{if} \qquad
C_V(T,V,\mathcal{Q}_i) = C_V(T,\mathcal{Q}_i)\,,
\end{equation}
where 
\begin{equation} \label{carnot}
\eta_{\rm Carnot}=1-\frac{T_{\rm c}}{T_{\rm h}}\,.
\end{equation}

\begin{theorem}
An ideal regenerative Stirling engine with fixed conserved charges \(\mathcal{Q}_i\) attains the Carnot efficiency if the constant-volume heat capacity is independent of volume.
\end{theorem}

\begin{proof}
\textcolor{black}{Under the hypothesis of the Proposition, the heat released
along one isochore equals the heat absorbed along the other at every
temperature. An ideal regenerator can therefore store and return this heat
reversibly, returning to its initial thermodynamic state with no external heat
exchange along the isochores. The efficiency of the regenerative Stirling cycle then reduces to}
\begin{equation} \label{intermediatestirlingefficiency}
\eta_{\rm regen}
= 1 - \frac{|Q_{3\rightarrow4}|}{|Q_{1\rightarrow2}|}
= 1 - \frac{T_{\rm c}(S_3-S_4)}{T_{\rm h}(S_2-S_1)} \,.
\end{equation}
Thus it suffices to show that \(S_3-S_4 = S_2-S_1\). It turns out that this entropy relation follows automatically from our criterion \eqref{criterionheatcapacity}. 
That is because, along reversible isochoric paths, one has
\begin{equation}
    S_2 - S_3 = \int_{T_3}^{T_2} \frac{C_V(T,V_2,\mathcal{Q}_i)}{T}\, dT\,,
    \qquad
    S_1 - S_4 = \int_{T_4}^{T_1} \frac{C_V(T,V_1,\mathcal{Q}_i)}{T}\, dT\,.
\end{equation}
If \(C_V\) is independent of volume, the integrands coincide as functions of \(T\), and since \(T_1=T_2\) and \(T_3=T_4\), the two integrals are equal. This implies
 $ 
    S_2 - S_3 = S_1 - S_4,
$ 
or equivalently \(S_3-S_4 = S_2-S_1\). Hence, the Stirling efficiency \eqref{intermediatestirlingefficiency} reduces to the Carnot efficiency.
\end{proof}

\noindent
A few comments are in order.

First, the assumption that all conserved charges \(\mathcal{Q}_i\) are held fixed throughout the cycle is essential. If any of the charges is allowed to vary, the relation
 $ 
    Q = \int C_V(T,V,\mathcal{Q}_i)\, dT
$
no longer holds. This is because, along an isochore, the entropy can change not only due to variations in temperature but also due to variations in the conserved charges. Infinitesimally, one has
\begin{equation}
    dS = \left(\frac{\partial S}{\partial T}\right)_{V,\mathcal{Q}_i} dT
    + \sum_i \left(\frac{\partial S}{\partial \mathcal{Q}_i}\right)_{T,V,\mathcal{Q}_{j\neq i}} d\mathcal{Q}_i\,.
\end{equation}
Using Clausius' relation and $C_V\equiv T(\partial S/\partial T)_{V,\mathcal Q_i}$, this gives
\begin{equation}
    Q
    = \int \left[
    C_V(T,V,\mathcal{Q}_i)\, dT
    + T \sum_i \left(\frac{\partial S}{\partial \mathcal{Q}_i}\right)_{T,V,\mathcal{Q}_{j\neq i}} d\mathcal{Q}_i
    \right].
\end{equation}
If the conserved charges are fixed, \(d\mathcal{Q}_i=0\), the second term vanishes and one recovers the relation used above. However, if any of the charges varies along the cycle, this additional contribution is generically nonzero, and the heat transfer is no longer determined solely by the temperature variation. Instead, it depends on the path in the space of thermodynamic variables, and the heat exchanged along the two isochores is no longer guaranteed to coincide. Consequently, their cancellation generically fails.

\begingroup\color{black}
Second, the criterion $C_V(T,V,\mathcal Q_i)=C_V(T,\mathcal Q_i)$ is sufficient but not necessary for pointwise matching on a particular cycle. Equality of the unweighted integrals,
\begin{equation}
    \int_{T_{\rm c}}^{T_{\rm h}} C_V(T,V_2,\mathcal{Q}_i)  dT
    =
    \int_{T_{\rm c}}^{T_{\rm h}} C_V(T,V_1,\mathcal{Q}_i)  dT\,,
\end{equation}
is sufficient only for the net energetic mismatch $Q_{\rm mis}$ to vanish. A necessary entropy balance condition is additionally
\begin{equation}
    \int_{T_{\rm c}}^{T_{\rm h}} \frac{C_V(T,V_2,\mathcal{Q}_i)}{T}  dT
    =
    \int_{T_{\rm c}}^{T_{\rm h}} \frac{C_V(T,V_1,\mathcal{Q}_i)}{T}  dT\,.
\end{equation}
For a reversible regenerator, a natural sufficient condition is the stronger
pointwise equality
$ 
C_V(T,V_2,\mathcal Q_i)=C_V(T,V_1,\mathcal Q_i)
$  
throughout the temperature interval. This ensures that the heat released during
cooling equals the heat required during heating at each temperature.
Condition~\eqref{criterionheatcapacity} guarantees this automatically.

Other routes to the Carnot limit are possible. Section~\ref{subsec:stirling_equal_to_carnot} realizes an asymptotic mechanism: the isochoric energy and entropy mismatches remain nonzero at finite electric potential $\tilde\Phi$ but become subleading relative to the isothermal heat as $\tilde\Phi\to\infty$.
\endgroup

Third, our result is fully consistent with   Carnot's theorem. According to (part of) Carnot's theorem, a reversible engine operating between two reservoirs at temperatures $T_{\rm h}$ and $T_{\rm c}$ attains the Carnot efficiency. For such an engine, all heat exchange with the reservoirs must occur isothermally, since any heat transfer at a different temperature would involve a finite temperature difference and hence generate entropy, violating reversibility. 
In a generic Stirling cycle, however, heat is also exchanged along the isochores, during which the temperature of the working substance varies continuously.  Implementing these isochoric steps reversibly requires a continuum of reservoirs, thereby violating the two reservoir assumption of Carnot's theorem. 

\textcolor{black}{The key point of our theorem is that condition~\eqref{criterionheatcapacity},
not merely $Q_{\rm mis}=0$, makes the isochoric heat exchanges match at every
temperature. The isochoric exchange is then entirely internal and reversible, and the only heat exchanged with external reservoirs occurs on the two isotherms. The cycle therefore satisfies the assumptions of Carnot's theorem.}

Finally, the criterion \eqref{criterionheatcapacity} can be verified explicitly 
for several familiar working substances. For classical ideal gases and Van der 
Waals fluids one has $C_V=\frac{f}{2}N$, independent of volume (at fixed $N$), 
so the condition is satisfied and Carnot efficiency is attained. This also 
explains the standard textbook result (see e.g. \cite{CengelBoles2014}) for the Stirling cycle, which typically assumes 
an ideal gas with perfect regeneration.  By contrast, the   quantum ideal gases and neutral
CFT working substances considered below do not satisfy this condition; their
deviations from Carnot efficiency follow from the explicit calculations of the
heat exchanged along the cycle.

\section{Stirling engine for   classical ideal gases}
\label{sec:idealgas}

We now compute the Stirling efficiency for a classical ideal gas with \(N\) particles and \(f\) degrees of freedom per particle, working in units where \(k_B=1\). We assume that the particle number \(N\) is conserved throughout the cycle, in accordance with the condition of fixed conserved charges discussed above. While this is a standard textbook calculation (see e.g. problem 4.21 in\ \cite{schroeder2000introduction}), it is instructive to revisit it in order to see the differences between the non-regenerative and regenerative cases. In particular, this example provides a concrete illustration of the theorem established above. 

\begin{figure}[t]
    \centering

    \begin{subfigure}{0.45\textwidth}
        \centering
        \includegraphics[width=\linewidth]{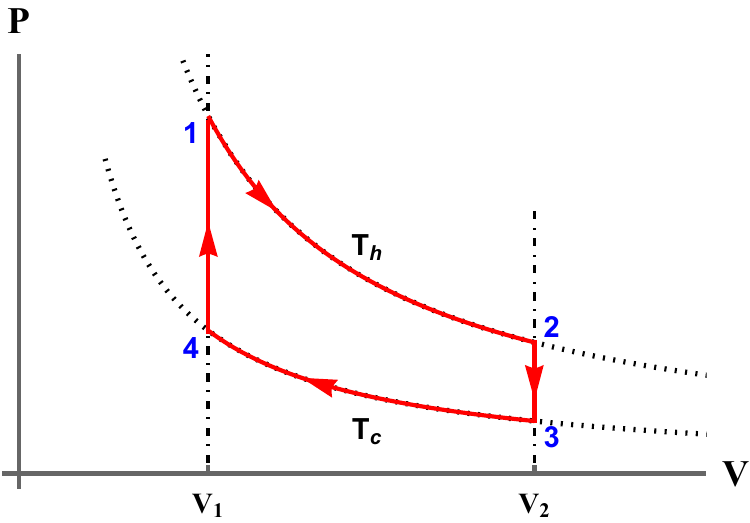}
        \caption{$PV$-diagram for a classical ideal gas. }
        \label{fig:sub1}
    \end{subfigure}
    \hfill
    \begin{subfigure}{0.45\textwidth}
        \centering
        \includegraphics[width=\linewidth]{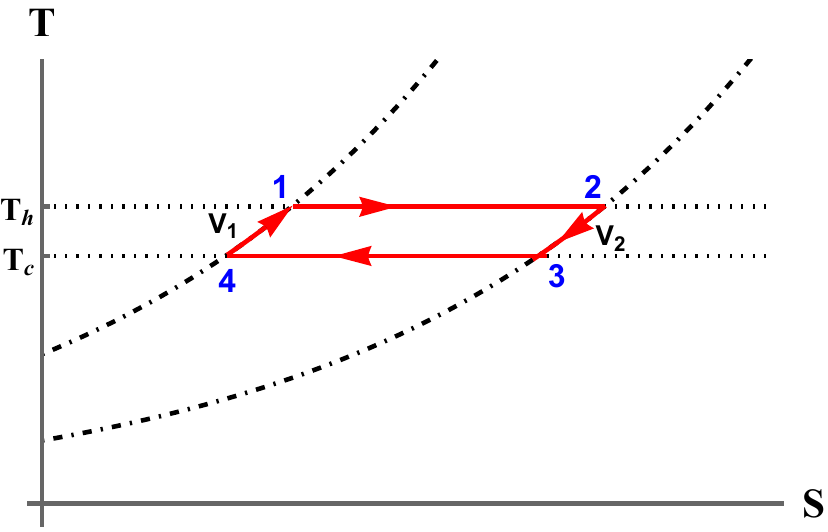}
        \caption{$TS$-diagram for a classical ideal gas. }
        \label{fig:sub2idealgas}
    \end{subfigure}

    \caption{Stirling cycle for a classical ideal gas. }
    \label{fig:main}
\end{figure}

In this case, the equation of state (a.k.a. ideal gas law) and internal energy are, respectively,
\begin{equation}
    PV = NT\,,
\qquad
    E = \frac{f}{2}NT\,. \label{idealgasenergy}
\end{equation}
These two relations are enough to evaluate explicitly the heat and work along each branch of the Stirling cycle. In particular, since the energy depends only on the temperature, one has \(\Delta E=0\) along an isotherm. Moreover,  it follows from the first law and Clausius' relation that the entropy change along an isotherm is $\Delta S = N \ln (V_2/V_1).$ \\

\noindent \textbf{In absence of regeneration.}
We begin with the two isothermal branches. Along \(1\to2\), the gas expands at temperature \(T_{\rm h}\), so the work is
\begin{equation}
    W_{1\to2}=\int_{V_1}^{V_2} P dV
    = \int_{V_1}^{V_2} \frac{NT_{\rm h}}{V} dV
    = NT_{\rm h}\ln\!\left(\frac{V_2}{V_1}\right)\,.
\end{equation}
Since \(\Delta E=0\) on an isotherm, the first law implies that the heat absorbed along this branch is equal to the work:
\begin{equation}
    Q^{1\to2}_{\rm in}=W_{1\to2}
    = NT_{\rm h}\ln\!\left(\frac{V_2}{V_1}\right)\,,
\end{equation}
which, equivalently, can be obtained from Clausius' relation $Q_{\rm in }^{1 \to 2}= T_{\rm h} \Delta S_{1 \to 2}$. 
Similarly, along the cold isotherm \(3\to4\),
\begin{equation}
    W_{3\to4}=\int_{V_3}^{V_4} P dV
    = \int_{V_2}^{V_1} \frac{NT_{\rm c}}{V} dV
    = -NT_{\rm c}\ln\!\left(\frac{V_2}{V_1}\right)\,.
\end{equation}
Again \(\Delta E=0\), so the magnitude of the heat rejected along this branch is
\begin{equation}
    Q^{3\to4}_{\rm out}=|W_{3\to4}|
    = NT_{\rm c}\ln\!\left(\frac{V_2}{V_1}\right).
\end{equation}
Along the isochores no work is done. The heat exchanged is therefore entirely given by the change in energy. Using \eqref{idealgasenergy}, one finds
\begin{align}
    Q^{2\to3}_{\rm out}=E_2-E_3
    &= \frac{f}{2}N(T_{\rm h}-T_{\rm c})\,, \\
    Q^{4\to1}_{\rm in}=E_1-E_4
   & = \frac{f}{2}N(T_{\rm h}-T_{\rm c})\,.
\end{align}
Thus the heat exchanged along the two isochores has the same magnitude.
Finally, by inserting the heat and work expressions into one of the   definitions of the Stirling efficiency,   equations \eqref{nonregen1} or \eqref{nonregen2},   one finds
\begin{equation} \label{nonregenidealgas}
    \eta_{\rm non\text{-}regen}^{\rm ideal\,gas}
    =
    \frac{(T_{\rm h}-T_{\rm c})\ln(V_2/V_1)}
    {T_{\rm h}\ln(V_2/V_1)+\frac{f}{2}(T_{\rm h}-T_{\rm c})}= \left( \frac{1}{\eta_{\rm Carnot}} + \frac{f}{2\ln (V_2/V_1)}\right)^{-1}\,.
\end{equation}
 The second term  between brackets on the right side is always positive, but is smaller for larger compression ratios $V_2/V_1$. Therefore, the efficiency is always smaller than the Carnot efficiency by an amount that is smaller when the compression ratio is larger.\\

\noindent \textbf{In presence of regeneration.}
In the ideal gas case, the heat capacity at constant volume is
\begin{equation} \label{classicalidealgasheat}
    C_V=\left(\frac{\partial E}{\partial T}\right)_{V,N}=\frac{f}{2}N\,,
\end{equation}
which is independent of volume and  temperature. As a result, the heat exchanged along an isochore depends only on the temperature difference \(T_{\rm h}-T_{\rm c}\), and the amount of heat released in the step \(2\to3\) is exactly equal to the amount of   heat absorbed in the step \(4\to1\):
\begin{equation}
    Q^{2\to3}_{\rm out} = Q^{4\to1}_{\rm in}
    = \frac{f}{2}N(T_{\rm h}-T_{\rm c}) \,.
\end{equation}
For an ideal regenerator, this heat is stored during \(2\to3\) and fully returned during \(4\to1\), so  
 $
    Q_{\text{mis}}=0.
$ 
The only external heat input and output is then that supplied along the hot and cold isotherms, respectively,
\begin{align}
    Q_{\rm in}&=Q^{1\to2}_{\rm in}
    = NT_{\rm h}\ln\!\left(\frac{V_2}{V_1}\right)\, , \\
    Q_{\rm out}&=Q^{3\to4}_{\rm out}=NT_{\rm c}\ln\left(\frac{V_2}{V_1}\right)\,,
\end{align}
while the total work done per cycle is unchanged,
\begin{equation}
    W = W_{1\to2}+W_{3\to4}=N(T_{\rm h}-T_{\rm c})\ln\!\left(\frac{V_2}{V_1}\right).
\end{equation}
\textcolor{black}{Hence, inserting these expressions into the regenerative efficiency gives the Carnot result. The formulas~\eqref{regen1}--\eqref{regen4} coincide because the isochoric
heat exchanges match at every temperature. Thus,}  
\begin{equation}\label{idealgasregenerative}
    \eta_{\rm regen}^{\rm ideal\,gas}
    = 1-\frac{T_{\rm c}}{T_{\rm h}}\,.
\end{equation}
The second term in \eqref{nonregenidealgas}, which arises from the heat exchange along the isochores, no longer contributes in the presence of an ideal regenerator, since this heat is entirely recycled internally and does not enter the external heat balance.

\section{Stirling engine for   Van der Waals fluids}
\label{sec:waals}

We now extend the discussion to a Van der Waals fluid. This provides a simple example of an interacting working substance, with the parameter \(a\) accounting for attractive interactions and \(b\) for excluded-volume effects associated with the finite size of the particles. The equation of state in any spacetime dimension $D>3$ is
\begin{equation}
    \left(P+\frac{aN^2}{V^2}\right)(V-bN)=NT \, ,
\end{equation}
and the internal energy is
\begin{equation}
    E=\frac{f}{2}NT-\frac{aN^2}{V}.
    \label{vdwenergy}
\end{equation}
Unlike the ideal gas case, for which $a=b=0$, the internal energy now depends explicitly on the volume, so in general \(\Delta E\neq 0\) along an isotherm. The entropy takes the form of a modified Sackur-Tetrode expression,
\begin{equation}
    S = N\ln(V-bN)+\frac{f}{2}N\ln T - N \ln N +N\times\text{const.},
\end{equation}
where \(b\) is the excluded volume per particle, so that \(bN\) is the total excluded volume and \(V-bN\) is the volume accessible to the gas. Therefore, along an isotherm (at fixed \(N\)),
\begin{equation}
    \Delta S = N\ln\!\left(\frac{V_2-bN}{V_1-bN}\right).
\end{equation}
The Van der Waals system has a critical point (second-order phase transition), which exists for all spacetime dimensions $D>3$ and occurs at
\begin{equation}
    T_{\rm crit} = \frac{8 a}{27 b}\,, \qquad V_{\rm crit} =  3Nb \, , \qquad P_{\rm crit} = \frac{a}{27 b^2}\,.
\end{equation}
In the canonical ensemble at fixed $(T,V,N)$, above the critical temperature, the Van der Waals system is in a   single homogeneous phase. Below the critical temperature, the system admits distinct liquid and gas phases, but in this ensemble they appear in coexistence, as a mixture   over a finite range of volumes (in contrast to the fixed $(T,P,N)$ ensemble, where there is a first-order liquid-gas phase transition for $T<T_{\rm crit}$). \\

\noindent \textbf{In absence of regeneration.}
We begin again with the hot isotherm \(1\to2\). Using the equation of state, the work is
\begin{align}
    W_{1\to2}
    =\int_{V_1}^{V_2}P dV =\int_{V_1}^{V_2}\left(\frac{NT_{\rm h}}{V-bN}-\frac{aN^2}{V^2}\right)dV  =NT_{\rm h}\ln\!\left(\frac{V_2-bN}{V_1-bN}\right)
      +aN^2\left(\frac{1}{V_2}-\frac{1}{V_1}\right)\,.
\end{align}
The change in internal energy along this branch is
\begin{equation}
    \Delta E_{1\to2}
    =E_2-E_1
    =-aN^2\left(\frac{1}{V_2}-\frac{1}{V_1}\right)\,.
\end{equation}
Hence, by the first law,
\begin{align}
    Q^{1\to2}_{\rm in}
    = \Delta E_{1\to2}+W_{1\to2}= NT_{\rm h}\ln\!\left(\frac{V_2-bN}{V_1-bN}\right)\,,
\end{align}
which also follows from
$ 
    Q^{1\to2}_{\rm in}=T_{\rm h}\Delta S_{1\to2}\,.
$ 
Thus the \(a\)-dependent terms cancel in the   heat input   along the isotherm.
Similarly, along the cold isotherm \(3\to4\),
\begin{align}
  W_{3 \to 4} &= - NT_{\rm c}\ln\!\left(\frac{V_2-bN}{V_1-bN}\right)
      - aN^2\left(\frac{1}{V_2}-\frac{1}{V_1}\right) \,, \\  Q^{3\to4}_{\rm out}
    &=NT_{\rm c}\ln\!\left(\frac{V_2-bN}{V_1-bN}\right)\,.
\end{align}
Along the isochores no work is done, so the heat exchange is entirely given by the change in internal energy. Equation \eqref{vdwenergy} yields
\begin{align}
    Q^{2\to3}_{\rm out}
    &=E_2-E_3
    =\frac{f}{2}N(T_{\rm h}-T_{\rm c})\,,\\
    Q^{4\to1}_{\rm in}
   &=E_1-E_4
    =\frac{f}{2}N(T_{\rm h}-T_{\rm c})\,.
\end{align}
\begin{figure}[t!]
    \centering

    \begin{subfigure}{0.45\textwidth}
        \centering
        \includegraphics[width=\linewidth]{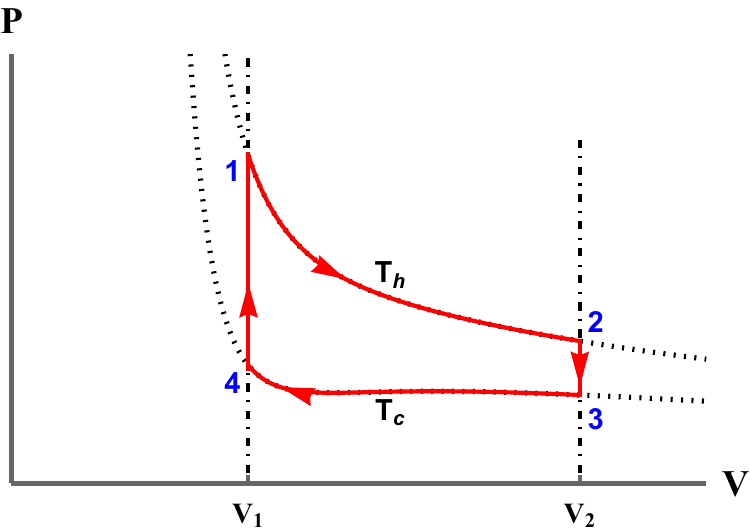}
        \caption{Here, the upper isotherm corresponds to a temperature above the critical point, $T_{\rm h} > T_{\text{crit}}$, while the lower isotherm is fixed at the critical temperature, $T_{\rm c } = T_{\rm crit}$.}
        \label{fig:sub1}
    \end{subfigure}
    \hfill
    \begin{subfigure}{0.45\textwidth}
        \centering
        \includegraphics[width=\linewidth]{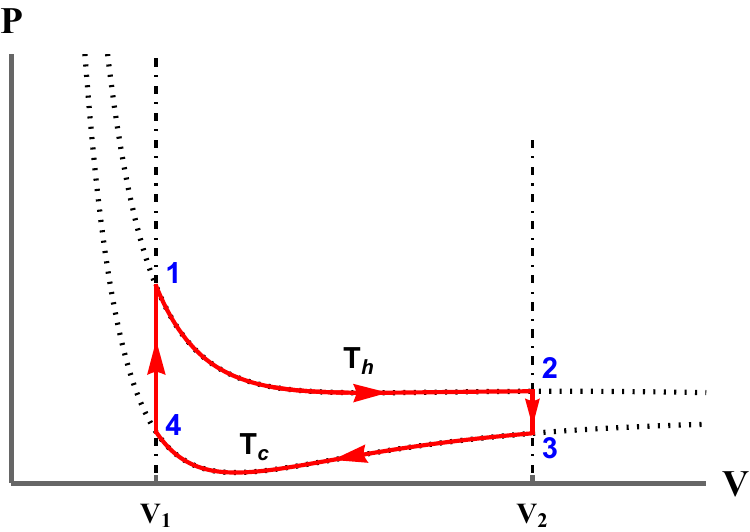}
        \caption{Here, the upper isotherm corresponds to the critical point, $T_{\rm h} = T_{\text{crit}}$, while the lower isotherm corresponds to a temperature below the critical point, $T_{\rm c} < T_{\text{crit}}$.}
        \label{fig:sub2}
    \end{subfigure}

    \caption{$PV$-diagrams for a Stirling cycle with a Van der Waals fluid working substance. This fluid has three qualitatively distinct isotherms in a $PV$-diagram: above, below and at the critical temperature $T_{\rm crit}$. Hence, there are 5 different ways in which one can operate the Stirling engine: (a) $T_{\rm h} > T_{\text{crit}}$, $T_{\rm c} = T_{\text{crit}}$ (b) $T_{\rm h} = T_{\text{crit}}$, $T_{\rm c} < T_{\text{crit}}$ (c) $T_{\rm h} > T_{\rm c} > T_{\text{crit}}$ (d) $T_{\rm c}<T_{\rm h}<T_{\text{crit}}$ (e) $T_{\rm h} > T_{\text{crit}} > T_{\rm c}$. We have presented the $PV$-diagrams for the first two cases. One can similarly obtain the other diagrams. The $TS$-diagram for a Van der Waals fluid is qualitatively identical to that of a classical ideal gas (see Figure \ref{fig:sub2idealgas}).}
    \label{fig:main}
\end{figure}
\noindent As in the ideal gas case, the two isochores exchange heat of equal magnitude.

Substituting into the definition of the efficiency,   equations \eqref{nonregen1} or \eqref{nonregen2}, gives
\begin{equation}
    \eta_{\rm non\text{-}regen}^{\rm VdW}
    =
    \frac{(T_{\rm h}-T_{\rm c})\ln\!\big((V_2-bN)/(V_1-bN)\big)}
    {T_{\rm h}\ln\!\big((V_2-bN)/(V_1-bN)\big)+\frac{f}{2}(T_{\rm h}-T_{\rm c})}\,,
\end{equation}
or, equivalently,
\begin{equation}
    \eta_{\rm non\text{-}regen}^{\rm VdW}
    =
    \left(
    \frac{1}{\eta_{\rm Carnot}}
    +\frac{f}{2\ln\!\big((V_2-bN)/(V_1-bN)\big)}
    \right)^{-1}\,.
\end{equation}
Therefore, the non-regenerative Stirling efficiency has the same form as for the ideal gas, except that the ordinary compression ratio is replaced by the ratio of available volumes, \((V_2-bN)/(V_1-bN)\). The efficiency does not depend on the interaction parameter $a$. \\

\noindent \textbf{In presence of regeneration.}
From \eqref{vdwenergy} one obtains that the heat capacity, 
\begin{equation}
    C_V=\left(\frac{\partial E}{\partial T}\right)_V=\frac{f}{2}N\,,
\end{equation}
  is   independent of volume. Therefore the criterion of the theorem is satisfied, and the heat exchanged along the two isochores has equal magnitude:
\begin{equation}
    Q^{2\to3}_{\rm out}=Q^{4\to1}_{\rm in}
    =\frac{f}{2}N(T_{\rm h}-T_{\rm c}).
\end{equation}
For an ideal regenerator, this heat is stored during \(2\to3\) and fully returned during \(4\to1\), so that
$ 
    Q_{\text{mis}}=0.
$  
Hence the only external heat input and output are those exchanged along the isotherms,
\begin{align}
    Q_{\rm in}
    &=NT_{\rm h}\ln\!\left(\frac{V_2-bN}{V_1-bN}\right),\\
    Q_{\rm out}
    &=NT_{\rm c}\ln\!\left(\frac{V_2-bN}{V_1-bN}\right),
\end{align}
while the work is unchanged,
\begin{equation}
    W=N(T_{\rm h}-T_{\rm c})\ln\!\left(\frac{V_2-bN}{V_1-bN}\right).
\end{equation}
It follows that the regenerative Stirling efficiency is
\begin{equation}
    \eta_{\rm regen}^{\rm VdW}
    =1-\frac{T_{\rm c}}{T_{\rm h}}.
\end{equation}
Thus, although the Van der Waals fluid includes interactions and excluded-volume effects, its regenerative Stirling efficiency is still equal to the Carnot efficiency. The reason is that \(C_V\) remains independent of volume, so the heat exchanged along the isochores is perfectly recycled by an ideal regenerator.

 \section{Stirling engine for quantum ideal gases}
\label{sec:quantumgases}

In this section we extend the analysis of the Stirling cycle to quantum ideal gases, with the aim of identifying deviations from the classical behavior. \textcolor{black}{Their volume-dependent heat capacity places them outside the sufficient condition of section~\ref{sec:stirlingtheorem}. The nonzero mismatch and the departure from Carnot efficiency are then established by the explicit calculations below; they do not follow from volume dependence alone.}

We consider quantum systems at fixed particle number throughout the cycle. For calculational convenience, however, we work in the grand canonical ensemble, treating the chemical potential as an independent variable and subsequently expressing it   in terms of the particle number. This procedure is justified by the equivalence of ensembles in the thermodynamic limit.

In our analysis we make two standard approximations. First, we replace the discrete sum over energy levels by a continuum integral using the density of states. Second, we neglect subleading finite-volume contributions. Both approximations are justified in the thermodynamic limit, where the energy level spacing becomes negligible and the system size is large compared to the thermal wavelength. There are many useful references on quantum ideal gases; in this work we have relied mostly on \cite{Pathria:1996hda}, but extended it to arbitrary dimensions.

\subsection{Ideal Bose gas}

We consider a non-relativistic gas of spin-zero, non-interacting bosons in $d$ spatial dimensions, confined to a volume $V$. 
The grand canonical partition function for an ideal Bose gas is given by 
\begin{equation}
\ln \mathcal{Z}
=
- \sum_{\epsilon} \ln \!\left(1-ze^{-\beta \epsilon}\right),
\qquad
\beta = \frac{1}{T}\,,
\qquad
z=e^{\beta \mu}\,,
\end{equation}
where the sum is over all single-particle energy eigenvalues $\epsilon$,   $\mu$ is the chemical potential and $z$ is the fugacity of the gas, satisfying $\mu \le0 $  and $0 \le z \le 1$.

The pressure and particle number, which we identify with its ensemble average,  follow from
\begin{equation}
P = \frac{T}{V} \ln \mathcal Z\,,
\qquad
N = z \left ( \frac{\partial}{\partial z} \ln \mathcal Z\right)_{T,V}
= \sum_{\epsilon} \frac{1}{z^{-1} e^{\beta \epsilon}-1}\,.
\end{equation}
Since the ground state $\epsilon=0$ may become macroscopically occupied, it must be separated before taking the thermodynamic limit. Writing
\begin{equation}
  N = \frac{z}{1-z} +  N_{\rm exc}   \equiv N_0 +   N_{\rm exc}   \,,
\label{eq:Bose-N-split}
\end{equation}
we identify $N_0$ as the ground-state occupation.
For the excited states, the sum may be replaced by an integral using the density of states,
\begin{equation}
\rho_d(\epsilon)
=
\frac{V}{\Gamma(d/2)}
\left(\frac{2\pi m}{h^2}\right)^{d/2}
\epsilon^{\frac{d}{2}-1}\,,
\end{equation}
where $m$ is the boson mass and $h$ is Planck's constant.
This yields
\begin{equation}
\frac{P}{T}
=
-\frac{1}{\Gamma(d/2)}
\left(\frac{2\pi m}{h^2}\right)^{d/2}
\int_0^\infty d\epsilon \,
\epsilon^{\frac{d}{2}-1}
\ln\!\left(1-ze^{-\beta\epsilon}\right)
+\frac{1}{V}\ln\!\left(\frac{1}{1-z}\right)\,,
\end{equation}
and
\begin{equation}
\frac{N}{V}
=
\frac{1}{\Gamma(d/2)}
\left(\frac{2\pi m}{h^2}\right)^{d/2}
\int_0^\infty d\epsilon \,
\frac{\epsilon^{\frac{d}{2}-1}}{z^{-1}e^{\beta \epsilon}-1}
+
\frac{1}{V}\frac{z}{1-z}\,.
\end{equation}
The last term in the pressure may be written as $\frac{1}{V}\ln(N_0+1)$ and therefore vanishes in the thermodynamic limit $V \to \infty$ at fixed density $N/V$, since it is at most of order $\mathcal O(\frac{1}{N}\ln N)$. The ground state hence does not contribute to the pressure in this limit.

Introducing $x=\beta\epsilon$, the above integrals can be expressed in terms of the Bose-Einstein functions
\begin{equation} \label{befunctions}
g_\nu(z)
=
\frac{1}{\Gamma(\nu)}
\int_0^\infty
\frac{x^{\nu-1}\,dx}{z^{-1}e^x-1} = z + \frac{z^2}{2^{\nu}} + \frac{z^3}{3^{\nu}} + \cdots\,,
\end{equation}
where the second equality follows from an expansion around $z=0.$
In terms of these functions, the particle number and pressure take the form
\begin{equation} \label{bosegasnumberpressure}
P
=
\frac{T}{\lambda_T^d} g_{d/2+1}(z)\,,\qquad 
\frac{N-N_0}{V} 
=
\frac{1}{\lambda_T^d} g_{d/2}(z)\,,
\end{equation}
where $\lambda_T = (h^2/2\pi m T)^{1/2}$ is the thermal de Broglie wavelength. Here, the ground-state contribution to the pressure has been neglected.
Using the equivalence of ensembles, we regard the fugacity as a function $z=z(T,V,N)$, implicitly determined by the second relation in \eqref{bosegasnumberpressure}. Although this relation cannot be inverted in closed form, it admits a systematic expansion for small $z$, using the expansion of the Bose-Einstein function \eqref{befunctions}.

The internal energy receives no contribution from the ground state and is given by
\begin{equation}
E
=
-\left(\frac{\partial }{\partial \beta} \ln \mathcal Z\right)_{z,V}
=
\frac{d}{2}\,\frac{VT}{\lambda_T^d}\, g_{d/2+1}(z)
=
\frac{d}{2}\,PV\,.
\end{equation}
Since we are interested in heat engines at fixed particle number, we express this in terms of $N$,
\begin{equation}
E
=
\frac{d}{2}\,TN
\left(1-\frac{N_0}{N}\right)
 \frac{g_{d/2+1}(z)}{g_{d/2}(z)}
\end{equation}
where the fugacity should be viewed as the function $z=z(T,V,N)$. 

The entropy follows from the Euler relation $S=(E+PV-\mu N)/T$, which holds in the thermodynamic limit, yielding
\begin{equation}
S
=
N\left[
\left(\frac{d}{2}+1\right)
\left(1-\frac{N_0}{N}\right)
\frac{g_{d/2+1}(z)}{g_{d/2}(z)}
-\ln z
\right]\,.
\end{equation}
The constant-volume heat capacity is obtained by differentiating
the energy at fixed $V$ and $N$. Using the second equation in \eqref{bosegasnumberpressure} to eliminate
$\partial z/\partial T$, one finds   for $z<1$
\begin{equation}
C_V= \left ( \frac{\partial E}{\partial T}\right)_{V,N}
=
N\,\frac{d}{2}\left[
\left(\frac{d}{2}+1\right)\frac{g_{1+d/2}(z)}{g_{d/2}(z)}
-
\frac{d}{2}\frac{g_{d/2}(z)}{g_{d/2-1}(z)}
\right].
\label{eq:Bose-CV}
\end{equation}
Thus $C_V$ depends explicitly on the fugacity $z$, and hence implicitly on the volume. In the classical regime $z\ll1$, one has $g_\nu(z)\approx z$, so
$C_V\approx \frac{d}{2}N$, independent of $z$ and of the volume. This agrees with the heat capacity \eqref{classicalidealgasheat} of a classical ideal gas with $f=d$.

A characteristic feature of the Bose gas for $d>2$ is the existence of a critical temperature for
Bose-Einstein condensation (BEC). This is defined by the condition that the excited states can
no longer accommodate all particles. Since $g_{d/2}(z)$ is maximal at $z=1$, the second  equation in \eqref{bosegasnumberpressure}
implies that the excited states saturate at
\begin{equation}
\frac{N}{V}
=
\frac{1}{\lambda_{T_{\rm crit}}^d} \zeta(d/2)\,,
\end{equation}
where   the Riemann zeta function $\zeta $ appears since $ g_{d/2}(1) = \zeta(d/2)$ for $d>2.$
Solving for $T$ gives the critical temperature
\begin{equation}
T_{\rm crit}
=
\frac{h^2}{2\pi m}
\left(
\frac{N/V}{\zeta(d/2)}
\right)^{2/d}\,.
\end{equation}
For $T> T_{\rm crit}$ one has $N_0 \ll N$, while for  $T<T_{\rm crit}$ a macroscopic fraction of particles occupies the ground state. Below and at the critical temperature, the fugacity reaches its maximal value $z=1$ in the thermodynamic limit, and the pressure, internal energy and entropy reduce to
\begin{equation} \label{criticalpressureetc}
P=AT^{d/2+1}\,,\quad E=\frac{d}{2}AVT^{d/2+1}\,, \quad \text{and} \quad S=\left(\frac{d}{2}+1\right)AVT^{d/2}\,,
\end{equation}
with $A=\left(\frac{2\pi m}{h^2}\right)^{d/2}\zeta\!\left(\frac{d}{2} +1\right).$
The constant-volume heat capacity is therefore
\begin{equation}
C_V=\left(\frac{\partial E}{\partial T}\right)_{V,N}
=
\frac d2\left(\frac d2+1\right)AVT^{d/2}=
N\,\frac d2\left(\frac d2+1\right)
\frac{\zeta(d/2+1)}{\zeta(d/2)}
\left(\frac{T}{T_{\rm crit}}\right)^{d/2}\,.
\label{eq:Bose-CV-BEC}
\end{equation}
 This depends on the volume, hence we do not expect the Carnot efficiency for a regenerative Stirling engine.

The order of the Bose-Einstein condensation phase transition may be inferred from the
behavior of the constant-volume heat capacity at $T=T_{\rm crit}$. At fixed $(V,N)$,
\[
C_V=-\,T\left(\frac{\partial^2 F}{\partial T^2}\right)_{V,N},
\]
so discontinuities in $C_V$ correspond to discontinuities in $\partial_T^2 F$.
Taking the limit $T\to T_{\rm crit}^{+}$ in \eqref{eq:Bose-CV}, one finds that for
$2<d\leq 4$ the function $g_{d/2-1}(z)$ diverges as $z\to 1$, so the second term in
\eqref{eq:Bose-CV} vanishes and $C_V$ matches continuously onto the condensed-phase
result \eqref{eq:Bose-CV-BEC}. Thus $\partial_T^2 F$ is continuous at $T_{\rm crit}$,
and the first non-analytic behavior appears only in a higher temperature derivative of $F$.
The transition is therefore continuous and higher than second order  (for $d=3$, it is third order).

For $d>4$, $g_{d/2-1}(1)=\zeta(d/2-1)$ is finite, so $C_V$ exhibits a finite jump at
$T_{\rm crit}$. Hence $\partial_T^2 F$ is discontinuous, while $\partial_T F=-S$
remains continuous, implying that $F$ is once but not twice differentiable at
$T_{\rm crit}$. The transition is therefore second order. The case $d=4$ is marginal, with logarithmic corrections in $C_V$ at $T_{\rm crit}$ arising
from the divergence of $g_1(z)$ as $z \to 1$.\\

\noindent \textbf{In absence of regeneration.}
As already noted, in the canonical ensemble the fugacity \(z\) depends on the temperature,   volume and particle number.  Consequently, at fixed particle number \(z\) varies along both isotherms and isochores.  Heat exchanged along the isotherms is given by \(Q=T\Delta S\), while along the isochores it is \(Q=\Delta E\).  Using the   expressions for the energy and entropy above, the heat exchanged on the cold and hot isochores is
\begin{equation}
\begin{aligned}
Q_{\rm out}^{2\rightarrow 3} &= E_{2}-E_{3}
= \frac{d}{2}N\Bigl(T_{\rm h}\,g_{22}-T_{\rm c}\,g_{12}\Bigr)\,,\\[0.2em]
Q_{\rm in}^{4\rightarrow 1} &= E_{1}-E_{4}
= \frac{d}{2}N\Bigl(T_{\rm h}\,g_{21}-T_{\rm c}\,g_{11}\Bigr)\,,
\end{aligned}
\end{equation}
where
\begin{equation}
g_{ij}\equiv \Bigl(1-\frac{N_{0}}{N}\Bigr)\frac{g_{d/2+1}(z_{ij})}{g_{d/2}(z_{ij})}\qquad \text{with}\qquad z_{ij} \equiv z(T_i, V_j,N)\,,
\end{equation}
and $T_1 \equiv T_{\rm c}$ and $T_2 \equiv T_{\rm h}$. 
Along the isotherms one finds 
\begin{equation}
\begin{aligned}
Q_{\rm in}^{1\rightarrow 2} &= T_{\rm h}\bigl(S_{2}-S_{1}\bigr)
= NT_{\rm h}\bigl(\tilde{g}_{22}-\tilde{g}_{21}\bigr),\\[0.2em]
Q_{\rm out}^{3\rightarrow 4} &= T_{\rm c}\bigl(S_{3}-S_{4}\bigr)
= NT_{\rm c}\bigl(\tilde{g}_{12}-\tilde{g}_{11}\bigr)\,,
\end{aligned}
\end{equation}
with \begin{equation}\tilde{g}_{ij}\equiv \bigl(\tfrac{d}{2}+1\bigr)\,g_{ij}-\ln z_{ij}\,.\end{equation}  Combining these heat flows yields the non-regenerative Stirling efficiency for an ideal Bose gas
\begin{equation} \label{efficiencyboseidealgasnonregen}
\eta_{\rm non\text{-}regen}^{\rm Bose}
=1-\frac{Q_{\rm out}^{3\rightarrow 4}+Q_{\rm out}^{2\rightarrow 3}}
{Q_{\rm in}^{1\rightarrow 2}+Q_{\rm in}^{4\rightarrow 1}}
=1-\frac{T_{\rm c}\bigl(\tilde{g}_{12}-\tilde{g}_{11}\bigr)
+\tfrac{d}{2}\,\bigl(T_{\rm h}\,g_{22}-T_{\rm c}\,g_{12}\bigr)}
{T_{\rm h}\bigl(\tilde{g}_{22}-\tilde{g}_{21}\bigr)
+\tfrac{d}{2}\,\bigl(T_{\rm h}\,g_{21}-T_{\rm c}\,g_{11}\bigr)}\,.
\end{equation}
This formula holds for all temperatures. Below the critical temperature, the ideal Bose gas is in the Bose-Einstein condensed phase. The condensate contributes via $N_0$ only to the particle number and not to the pressure, internal energy, or entropy, which are given by the relations \eqref{criticalpressureetc}. Substituting these relations  into the non-regenerative Stirling efficiency gives
\begin{equation}
\eta_{\rm non\text{-}regen}^{\rm BEC}
=
\frac{\big(T_{\rm h }^{d/2+1}-T_{\rm c}^{d/2+1}\big)(V_2-V_1)}
{\left ( \frac{d}{2}+1 \right)T_{\rm h }^{d/2+1}(V_2-V_1)+\frac{d}{2}V_1\big(T_{\rm h }^{d/2+1}-T_{\rm c}^{1+d/2}\big)}\,.
\end{equation}
When the temperature is well above the critical value, $T\gg T_{\rm crit}$, the fugacity is small, $z\ll1$, so that $g_\nu(z)\approx z$ and $N_0/N\approx  0$. It follows that $g_{ij}\approx  1$ and $\tilde g_{ij}\approx  \frac d2+1-\ln z_{ij}$. The particle density reduces to $N/V\approx z/\lambda_T^d$, so that along an isotherm $z\propto V^{-1}$ and hence $\tilde g_{i2}-\tilde g_{i1}=\ln(V_2/V_1)$. Substituting into the efficiency \eqref{efficiencyboseidealgasnonregen}, the expression reduces to the classical ideal gas result \eqref{nonregenidealgas}, with the number of degrees of freedom given by $f=d$, corresponding to the number of translational   degrees of freedom  in $d$ spatial dimensions. \\

\noindent \textbf{In presence of regeneration.}
With a regenerator, heat released on the cold isochore can be stored and reused on the subsequent hot isochore. For an ideal Bose gas, however, the amounts of heat released during the cooling step \(2\rightarrow 3\) and required during the heating step \(4\rightarrow 1\) are not equal, because \(g_{22}\neq g_{21}\) when \(z\) depends on both \(T\) and \(V\).
A direct evaluation shows that, in the normal phase \(0<z<1\) above the critical temperature,
\begin{equation}\label{inequalitybose}
Q^{4\to1}_{\rm in}>Q^{2\to3}_{\rm out}\,.
\end{equation}
This follows from
\begin{equation} \label{eq:mismatchheatcapacity}
Q^{4\to1}_{\rm in}-Q^{2\to3}_{\rm out}
=
\int_{T_{\rm c}}^{T_{\rm h}}\!\big(C_V(T,V_1,N)-C_V(T,V_2,N)\big)\,dT.
\end{equation}
At fixed \((T,N)\), the fugacity is determined by \(N/V=\lambda_T^{-d}g_{d/2}(z)\), so that
\(V_1<V_2\) implies \(z(T,V_1,N)>z(T,V_2,N)\). In the normal phase the heat capacity increases
with the degree of quantum degeneracy, i.e.\ with \(z\), hence
$C_V(T,V_1,N)>C_V(T,V_2,N)$ and therefore $Q^{4\to1}_{\rm in}>Q^{2\to3}_{\rm out}$. \textcolor{black}{This pointwise inequality also verifies the fixed-sign condition underlying the use of $Q_{\rm mis}$ in this phase.}

In the classical regime \(z\ll1\), we saw above that the heat capacity is independent of both \(z\) and \(V\). The integrand in \eqref{eq:mismatchheatcapacity} therefore vanishes and the two heat exchanges coincide,
\(Q^{4\to1}_{\rm in}=Q^{2\to3}_{\rm out}\), implying \(Q_{\rm mis}=0\).
By contrast, below the critical temperature the fugacity saturates at \(z=1\), and the internal energy scales as
\(E=\frac{d}{2}AVT^{1+d/2}\). It follows that
\(
Q^{2\to3}_{\rm out}\propto V_2
\)
and
\(
Q^{4\to1}_{\rm in}\propto V_1
\),
so that for \(V_2>V_1\) the inequality is reversed,
\(Q^{4\to1}_{\rm in}<Q^{2\to3}_{\rm out}\).

\textcolor{black}{The formulae below apply directly when the complete isochoric temperature interval lies in a regime with one fixed sign of the local mismatch (normal or condensed). If an isochore crosses a temperature at which that sign reverses, the interval must be split there and the local surplus and deficit retained separately, as explained in section~\ref{sec:2.1}; the single net quantity $Q_{\rm mis}$ is then not sufficient for a
reversible treatment.}

\begin{figure}[t]
    \centering

    \begin{subfigure}[t]{0.45\textwidth}
        \centering
        \includegraphics[width=\linewidth]{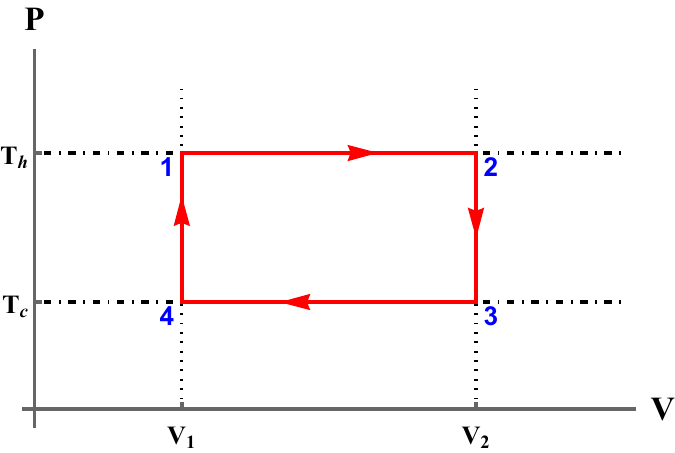}
        \caption{$PV$-diagram for BEC. Note that the isotherms coincide with the isobars, as follows from \eqref{criticalpressureetc}.}
        \label{fig:sub1}
    \end{subfigure}
    \hfill
    \begin{subfigure}[t]{0.45\textwidth}
        \centering
        \includegraphics[width=\linewidth]{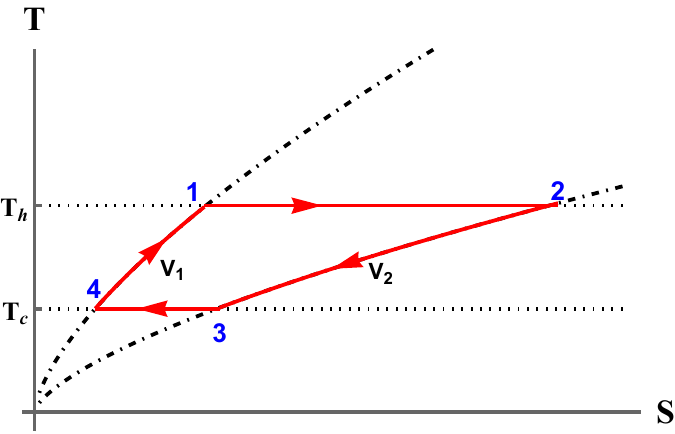}
        \caption{$TS$-diagram for BEC. }
        \label{fig:sub2idealgas}
    \end{subfigure}

    \caption{Stirling cycle for a Bose-Einstein condensate.}
    \label{fig:main}
\end{figure}

In the normal phase, the regenerator therefore cannot store enough heat to raise the gas back to \(T_{\rm h}\) on its own, and an additional amount
\begin{equation}
Q_{\rm mis}=|Q_{\rm in}^{4\rightarrow 1}-Q_{\rm out}^{2\rightarrow 3}|
\end{equation}
must be supplied externally at the corresponding temperatures. In the presence of a regenerator, the isothermal heat exchanges remain unchanged, so \(Q_{\rm out}=Q_{\rm out}^{3\rightarrow 4}\) and \(Q_{\rm in}=Q_{\rm in}^{1\rightarrow 2}+Q_{\rm mis}\). Thus the regenerative Stirling efficiency for an ideal Bose gas becomes
\begin{equation}
\eta_{\rm regen}^{\rm Bose}
=1-\frac{Q_{\rm out}^{3\rightarrow 4}}
{Q_{\rm in}^{1\rightarrow 2}+Q_{\rm mis}}
=1-\frac{T_{\rm c}\bigl(\tilde{g}_{12}-\tilde{g}_{11}\bigr)}
{T_{\rm h}\bigl(\tilde{g}_{22}-\tilde{g}_{21}\bigr)
+\tfrac{d}{2}\left[\bigl(T_{\rm h}\,g_{21}-T_{\rm c}\,g_{11}\bigr)
-\bigl(T_{\rm h}\,g_{22}-T_{\rm c}\,g_{12}\bigr)\right]}\,.
\end{equation}
The bracketed term in the denominator reflects the mismatch between the heat released and absorbed along the isochores. \textcolor{black}{The explicit calculation of the isochoric heat exchanges shows that the
regenerative efficiency remains below Carnot at finite quantum degeneracy. In the classical regime $T\gg T_{\rm crit}$, one has $z\ll1$ and $C_V$ becomes volume independent to leading order, so the local isochoric mismatch vanishes pointwise; correspondingly, $Q_{\rm mis}\approx0$ and $\tilde g_{12}-\tilde g_{11}\approx\tilde g_{22}-\tilde g_{21}$.} In this regime the efficiency reduces to the classical ideal gas result \eqref{idealgasregenerative} obtained earlier.

Below and at the critical temperature, where $z=1$, the energy is linear in the
volume, $E\propto V T^{1+d/2}$. Along an isochore $Q=\Delta E \propto V$, and since
the temperature difference is the same on both branches, the heat exchanged
differs only by the volume factor. Because $V_2>V_1$, the heat released during
cooling exceeds that required during heating, so the mismatch contributes to the
heat expelled, and the regenerative efficiency reduces to
\begin{equation}
\eta_{\rm regen}^{\rm BEC}
=1-\frac{Q_{\rm out}^{3\rightarrow 4}+Q_{\rm mis}}
{Q_{\rm in}^{1\rightarrow 2}}
=
\frac{1-\left(T_{\rm c}/T_{\rm h}\right)^{1+d/2}}{1+d/2}\,.
\end{equation}
Thus,   in the presence of a regenerator, the Bose-Einstein condensed phase does not attain the Carnot efficiency, as expected from \eqref{eq:Bose-CV-BEC}.

 \subsection{Ideal Fermi gas}

We now consider a non-relativistic ideal Fermi gas in $d$ spatial dimensions, whose grand canonical partition function is defined by
\begin{equation}
\ln \mathcal Z= \sum_{\epsilon} \ln\!\left(1 + z e^{-\beta \epsilon}\right)\,,
\qquad
z = e^{\beta \mu}.
\end{equation}
In contrast to the Bose gas, the fugacity is not bounded from above, i.e. $0\le z < \infty$. This reflects the Pauli exclusion principle, which limits the occupation number of each single-particle state to at most unity. As a result, even as $z\to\infty$ the occupation numbers remain finite and no macroscopic population of a single state occurs, so no separate ground-state contribution is required.

Passing to the continuum using the same density of states as in the Bose gas case, the thermodynamic quantities can be expressed in terms of the Fermi-Dirac functions  
\begin{equation}
f_\nu(z) = \frac{1}{\Gamma(\nu)} \int_0^\infty dx \, \frac{x^{\nu-1}}{z^{-1} e^x + 1}
= z - \frac{z^2}{2^\nu} + \frac{z^3}{3^\nu} - \cdots\,,
\end{equation}
where the second equality is the expansion around $z=0$.
In particular, the pressure and particle number in the thermodynamic limit are
\begin{equation} \label{pressurenumberfermi}
P = \frac{T}{\lambda_T^d} f_{d/2+1}(z), 
\qquad 
\frac{N}{V} = \frac{1}{\lambda_T^d} f_{d/2}(z)\,.
\end{equation}
Switching to the canonical ensemble, the second relation determines the fugacity as a function of $(T,V,N)$, i.e. $z=z(T,V,N)$.

The internal energy and entropy take the form
\begin{align}
E &= \frac{d}{2} N T \frac{f_{d/2+1}(z)}{f_{d/2}(z)} = \frac{d}{2}PV, \\
S &= N\left[\left(\frac{d}{2}+1\right)\frac{f_{d/2+1}(z)}{f_{d/2}(z)} - \ln z \right]\,.
\end{align}
Further, the constant-volume heat capacity depends on the fugacity and therefore implicitly on the volume,
\begin{equation} \label{heatcapacityfermi}
C_V
=
N\frac{d}{2}
\left[
\left(\frac{d}{2}+1\right)\frac{f_{d/2+1}(z)}{f_{d/2}(z)}
-
\frac{d}{2}\frac{f_{d/2}(z)}{f_{d/2-1}(z)}
\right]\,.
\end{equation}
This expression is valid for all temperatures.

Although there is no condensation phase for a Fermi gas, there is a quantum degeneracy regime characterized by the Fermi temperature $T_F$. This scale is defined from the zero-temperature limit of the    particle number relation.
At finite temperature, the occupation of a single-particle state of energy $\epsilon$ is given by the Fermi-Dirac distribution
\begin{equation}
n(\epsilon) = \frac{1}{e^{\beta(\epsilon-\mu)}+1}\,.
\end{equation}
In the limit $T \to 0$, this distribution becomes a step function: for $\epsilon < \mu$ one has $n(\epsilon)\to 1$, while for $\epsilon > \mu$ one has $n(\epsilon)\to 0$. The system therefore fills all available states up to a maximum energy $\varepsilon_F = \mu(T=0)$, known as the Fermi energy.
The particle number is then obtained by integrating over all occupied states,
\begin{equation}
\frac{N}{V}
=
\int_0^\infty d\epsilon \, \rho_d(\epsilon)\, n(\epsilon)
=
\int_0^{\varepsilon_F} d\epsilon \, \rho_d(\epsilon)\,,
\end{equation}
which determines the Fermi energy as
\begin{equation}
\varepsilon_F
=
\frac{h^2}{2m}
\left[
\frac{\Gamma(d/2+1)}{\pi^{d/2}}
\frac{N}{V}
\right]^{2/d},
\qquad
T_F \equiv \varepsilon_F\,.
\end{equation}
The degenerate regime corresponds to $T\ll T_F$, or equivalently $z\gg 1$. In this regime the Fermi-Dirac distribution remains close to a step function, and only states within an energy window of order $T$ around the Fermi energy are thermally excited.  Performing a low-temperature expansion then yields
\begin{equation}
E_0=\frac{d}{d+2}N\varepsilon_F\,,
\end{equation}
and the low-temperature energy becomes
\begin{equation}
E
=
E_0+\frac{\pi^2}{12}dN\frac{T^2}{T_F}
+\mathcal O\!\left(\frac{T^4}{T_F^3}\right)\,.
\end{equation}
It follows that the heat capacity at low temperatures is
\begin{equation} \label{heatcapacityfermideg}
C_V
=
\frac{\partial E}{\partial T}
=
\frac{\pi^2}{6}dN\frac{T}{T_F}
+\mathcal O\!\left(\frac{T^3}{T_F^3}\right)\,.
\end{equation}
Since $T_F\propto (N/V)^{2/d}$, the heat capacity depends explicitly on  the volume, even for $T\ll T_F$. \\

\noindent \textbf{In absence of regeneration.}
As in the Bose gas case, heat exchange along the isotherms is given by $Q=T\Delta S$, while along the isochores it is given by $Q=\Delta E$. Since the fugacity depends on both temperature and volume, it varies along all segments of the cycle.

It is convenient to introduce the shorthand
\begin{equation}
f_{ij} = \frac{f_{d/2+1}(z_{ij})}{f_{d/2}(z_{ij})},
\qquad
\tilde f_{ij} = \left(\frac{d}{2}+1\right) f_{ij} - \ln z_{ij}\,,
\end{equation}
where $z_{ij} \equiv z(T_i,V_j,N)$ and   $T_1\equiv T_{\rm c}$ and $T_2\equiv T_{\rm h}$. In terms of these quantities, the heat exchanged along each branch can be written in a compact form, and combining the contributions from the isothermal and isochoric segments yields
\begin{equation}
\eta_{\rm non\text{-}regen}^{\rm Fermi}
=
1 -
\frac{
T_{\rm c}(\tilde f_{12} - \tilde f_{11})
+
\frac{d}{2}(T_{\rm h} f_{22} - T_{\rm c} f_{12})
}{
T_{\rm h}(\tilde f_{22} - \tilde f_{21})
+
\frac{d}{2}(T_{\rm h} f_{21} - T_{\rm c} f_{11})
}\,.
\end{equation}
This expression is structurally identical to the Bose gas result, with the replacement of Bose-Einstein functions by Fermi-Dirac functions. In the classical limit $z \ll 1$ it reduces to the efficiency of a classical ideal gas, cf. \eqref{nonregenidealgas}.\\

\noindent \textbf{In presence of regeneration.}
In the presence of regeneration, however, there is a key difference compared to the case of an ideal Bose gas. As in the Bose gas case, the heat exchanges along the isochores do not match: $Q^{\text{out}}_{2 \rightarrow 3} \neq Q^{\text{in}}_{4 \rightarrow 1}$. \textcolor{black}{For the ideal Fermi gas the Bose gas inequality~\eqref{inequalitybose} is reversed. 
At fixed $(T,N)$, decreasing the volume increases the fugacity and the degree
of degeneracy. Numerical evaluation of the exact heat capacity
expression~\eqref{heatcapacityfermi} for integer dimensions
$2\leq d\leq10$ shows that $C_V$ decreases with fugacity. This is also consistent with the
degenerate-limit result $C_V\propto T/T_F$, since
$T_F\propto(N/V)^{2/d}$. Thus,   $V_1<V_2$ implies
\begin{equation}
C_V(T,V_1,N)<C_V(T,V_2,N)\,.
\end{equation}
The heat absorbed along the isochore at $V_1$ is consequently smaller than
the heat released along the isochore at $V_2$. The local mismatch has the same
sign throughout the temperature interval, and hence
\begin{equation}
Q_{\rm in}^{4\to1}<Q_{\rm out}^{2\to3}\,.
\end{equation}
Thus, the regenerator stores more heat during isochoric cooling than is
required during isochoric heating, and the excess
$
Q_{\rm mis}
=
Q_{\rm out}^{2\to3}-Q_{\rm in}^{4\to1}
 $
must be expelled externally at the temperatures at which it occurs.}

Taking this mismatch into account, one finds that the efficiency becomes
\begin{equation}
\eta_{\rm regen}^{\rm Fermi}
=
1 -
\frac{
T_{\rm c}(\tilde f_{12} - \tilde f_{11})
+
\frac{d}{2}
\Big[
(T_{\rm h} f_{22} - T_{\rm c} f_{12})
-
(T_{\rm h} f_{21} - T_{\rm c} f_{11})
\Big]
}{
T_{\rm h}(\tilde f_{22} - \tilde f_{21})
}\,.
\end{equation}
In the  classical regime $T \gg T_F$, one has $z \ll 1$ and $f_\nu(z)\approx z$. The dependence on the fugacity drops out, the heat capacity becomes independent of the volume, and the result  reduces  to the Carnot efficiency of a classical ideal gas.

\section{Stirling engine for thermal CFTs}
\label{sec:thermalcfts}

Thermal states of conformal field theories provide a natural class of working
substances for heat engines. On the one hand, they describe relativistic
many-body systems at fixed points, whose equilibrium thermodynamics is
strongly constrained by conformal symmetry. On the other hand, in holographic settings
these states admit a dual black hole description, so that heat engines
constructed from CFT thermal states probe black hole thermodynamics from the
boundary perspective.

For a CFT, scale invariance implies
that the stress-energy tensor is traceless:
$ 
T^\mu{}_\mu=0\,.
$ For a thermal CFT in global equilibrium on a spatially homogeneous geometry,
the stress-energy tensor is isotropic and the pressure is spatially constant: $T_{ij}= P \delta_{ij}$. The tracelessness condition then 
 gives $-\epsilon +(D-1)P=0$, where $\epsilon=E/V$ is the energy
density and $D$ is the number of spacetime dimensions. Therefore, 
\begin{equation}  
E=(D-1)PV\, .
\label{7.1}
\end{equation}
This is the conformal equation of state. It is sufficient to determine the efficiency of
 several standard heat engines, such as  the Otto, Diesel and Brayton engines \cite{Lilani:2025gzt}. For  the Stirling cycle the   efficiency is also fixed by the conformal equation of state, but its expression does not depend solely on the characteristic parameters $ T$ and $V$, as we will explain below. 
\\

\noindent \textbf{In absence of regeneration.}
We first consider the CFT Stirling cycle without a regenerator, for which the efficiency was computed   in \cite{Lilani:2025gzt}. We review this derivation here.  
Along the isothermal branches, Clausius'
relation gives
\begin{align}
Q^{1\to 2}_{\rm in} &= T_{\rm h}(S_2-S_1)\,, \label{7.2}\\
Q^{3\to 4}_{\rm out} &= T_{\rm c}(S_3-S_4)\,. \label{7.3}
\end{align}
Along the isochores no work is done, so the heat exchanged equals  the change in
internal energy:
\begin{align}
Q^{2\to 3}_{\rm out} &= E_2-E_3=(D-1)V_2(P_2-P_3)\,, \label{7.4}\\
Q^{4\to 1}_{\rm in} &= E_1-E_4=(D-1)V_1(P_1-P_4)\,, \label{7.5}
\end{align}
where the conformal equation of state \eqref{7.1} was used in the last equality. 
Inserting these heat exchanges into formula \eqref{nonregen2} for  the non-regenerative Stirling efficiency yields
\begin{equation}
\eta_{\rm non\text{-}regen}^{\rm CFT}
=
1-
\frac{T_{\rm c}(S_3-S_4)+(D-1)V_2(P_2-P_3)}
{T_{\rm h}(S_2-S_1)+(D-1)V_1(P_1-P_4)} .
\label{7.10}
\end{equation}
This formula is completely general for a CFT, but it is not yet expressed purely in terms
of the characteristic variables of the Stirling cycle, namely the temperatures and volumes.
To do that, we need explicit expressions for $S(T,V)$ and $P(T,V)$.

We therefore place the CFT on a round sphere of radius $R$, with volume
$
V=\Omega_{k,D-1}R^{D-1}\,.
$
For a CFT on a sphere, scale invariance implies that dimensionless thermodynamic
quantities depend on $T$ and $V$ only through the combination $TR$. In particular,
the canonical (Helmholtz) free energy admits the high-temperature (or large-volume) expansion \cite{Kutasov:2000td}
\begin{equation}
-FR
=
a_D(2\pi TR)^D+a_{D-2}(2\pi TR)^{D-2}
+a_{D-4}(2\pi TR)^{D-4}+\cdots .
\label{7.12}
\end{equation}
Here, the coefficients $a_D, a_{D-2}$ etc. are constants, but in general they may depend on the conserved charges $\mathcal Q_i
$. At strong coupling this expansion generally contains an infinite series in powers of
$(TR)^{-1}$, while at zero coupling the coefficients of the negative powers vanish.
Non-perturbative terms are exponentially suppressed for $TR\gg1$ and will be neglected.

The entropy follows from differentiating \eqref{7.12} with respect to the temperature at fixed volume
\begin{equation}  
S=-\left(\frac{\partial F}{\partial T}\right)_V = \frac{a_DD(2\pi)^D}{\Omega_{k,D-1}}\,T^{D-1}V  \xi (T,V) \,,
\label{7.13}
\end{equation}
where we defined
\begin{equation}
  \xi(T, V)
\equiv
1+\frac{a_{D-2}(D-2)}{a_DD(2\pi)^2T^2}
\left(\frac{\Omega_{k,D-1}}{V}\right)^{\frac{2}{D-1}}
+\frac{a_{D-4}(D-4)}{a_DD(2\pi)^4T^4}
\left(\frac{\Omega_{k,D-1}}{V}\right)^{\frac{4}{D-1}}
+\cdots .
\label{7.15}
\end{equation}
The entropy differences are then equal to
\begin{align}
S_2-S_1
&=
\frac{a_DD(2\pi)^D}{\Omega_{k,D-1}}\,T_{\rm h}^{D-1}
\left(V_2\xi_{22}-V_1\xi_{21}\right)\,, \label{7.16}\\
S_3-S_4
&=
\frac{a_DD(2\pi)^D}{\Omega_{k,D-1}}\,T_{\rm c}^{D-1}
\left(V_2\xi_{12}-V_1\xi_{11}\right)\,, \label{7.17}
\end{align}
where $\xi_{ij} \equiv \xi (T_i, V_j)\,$ and $T_1 \equiv T_{\rm c}$ and $T_2\equiv T_{\rm h}$. 
Next, the pressure is obtained from differentiating the free energy with respect to the volume at fixed temperature
\begin{equation}
P=-\left(\frac{\partial F}{\partial V}\right)_T = \frac{a_D(2\pi)^D}{\Omega_{k,D-1}}\,T^D \chi(T,V)\, ,
\label{7.18}
\end{equation}
where we defined
\begin{equation}
\chi (T,V)
\equiv
1+\frac{a_{D-2}(D-3)}{a_D(D-1)(2\pi)^2T^2}
\left(\frac{\Omega_{k,D-1}}{V}\right)^{\frac{2}{D-1}}
+\frac{a_{D-4}(D-5)}{a_D(D-1)(2\pi)^4T^4}
\left(\frac{\Omega_{k,D-1}}{V}\right)^{\frac{4}{D-1}}
+\cdots .
\label{7.21}
\end{equation}
Hence,
\begin{align}
P_2-P_3
&=
\frac{a_D(2\pi)^D}{\Omega_{k,D-1}}
\left(T_{\rm h}^D\chi_{22}-T_{\rm c}^D\chi_{12}\right)\,, \label{7.22}\\
P_1-P_4
&=
\frac{a_D(2\pi)^D}{\Omega_{k,D-1}}
\left(T_{\rm h}^D\chi_{21}-T_{\rm c}^D\chi_{11}\right)\,. \label{7.23}
\end{align}
Substituting \eqref{7.16}-\eqref{7.23} into \eqref{7.10}, the overall prefactor cancels,
and the efficiency becomes
\begin{equation}
\eta_{\rm non\text{-}regen}^{\rm CFT}
=
1-
\frac{
T_{\rm c}^D\left(V_2\xi_{12}-V_1\xi_{11}\right)
+\frac{D-1}{D}V_2\left(T_{\rm h}^D\chi_{22}-T_{\rm c}^D\chi_{12}\right)
}{
T_{\rm h}^D\left(V_2\xi_{22}-V_1\xi_{21}\right)
+\frac{D-1}{D}V_1\left(T_{\rm h}^D\chi_{21}-T_{\rm c}^D\chi_{11}\right)
} \,.
\label{7.24}
\end{equation}
This is the high-temperature or large-volume expansion of the non-regenerative Stirling
efficiency for a general CFT on a round sphere.
In the strict planar limit $TR\to\infty$, the finite-size corrections disappear, so that
$\xi_{ij}\to1$ and $\chi_{ij}\to1$. The efficiency then reduces to
\begin{equation}
\eta_{\rm non\text{-}regen}^{\rm CFT, plane}
=
1-
\frac{
T_{\rm c}^D(V_2-V_1)+\frac{D-1}{D}V_2(T_{\rm h}^D-T_{\rm c}^D)
}{
T_{\rm h}^D(V_2-V_1)+\frac{D-1}{D}V_1(T_{\rm h}^D-T_{\rm c}^D)
} \,.
\label{7.25}
\end{equation}\\

\noindent \textbf{In presence of regeneration.}
We next consider the same thermal CFT, but now in the presence of a regenerator.
We first need to compute the mismatch between the
heat released during   isochoric cooling, $Q_{\rm out}^{2 \to 3}$ in \eqref{7.2},  and the heat absorbed during  isochoric
heating, $Q_{\rm in}^{4 \to 1}$ in \eqref{7.3}. The general expression for the regenerative Stirling efficiency
depends on which of these two quantities is larger. In the present case, the
high-temperature or large-volume expansion shows that
\begin{equation}
Q^{2\to3}_{\rm out}>Q^{4\to1}_{\rm in}\,,
\end{equation}
\textcolor{black}{since the leading local heat capacity contribution scales as $VT^{D-1}$ and $V_2>V_1$. Thus the local mismatch has a fixed sign in the high-temperature or large-volume regime: the regenerator stores more heat during $2\to3$ than is needed during $4\to1$. The excess is rejected externally at the corresponding temperatures, so that}
\begin{equation}
\eta_{\rm regen} 
=
1-\frac{T_{\rm c}(S_3-S_4)+Q_{\text{mis}}}{T_{\rm h}(S_2-S_1)}\,,
\qquad
Q_{\text{mis}}=Q^{2\to3}_{\rm out}-Q^{4\to1}_{\rm in}\,.
\label{7.27a}
\end{equation}
Using the conformal equation of state, this becomes
\begin{equation}
\eta_{\rm regen}^{\rm CFT}
=
1-\frac{T_{\rm c}(S_3-S_4)+(D-1)\left[V_2(P_2-P_3)-V_1(P_1-P_4)\right]}
{T_{\rm h }(S_2-S_1)}\,.
\label{7.27b}
\end{equation}
We now specialize to the high-temperature expansion introduced above. Using the expressions
for $S(T,V)$ and $P(T,V)$, we obtain
\begin{equation}
Q_{\text{mis}} =
\frac{a_D(D-1)(2\pi)^D}{\Omega_{k,D-1}}
\left[
V_2\left(T_{\rm h}^D\chi_{22}-T_{\rm c}^D\chi_{12}\right)
-
V_1\left(T_{\rm h}^D\chi_{21}-T_{\rm c}^D\chi_{11}\right)
\right]\,,
\label{7.28}
\end{equation}
and therefore
\begin{equation}
\eta_{\rm regen}^{\rm CFT}
=
1-
\frac{
T_{\rm c}^D\left(V_2\xi_{12}-V_1\xi_{11}\right)
+\frac{D-1}{D}
\left[
V_2\left(T_{\rm h }^D\chi_{22}-T_{\rm c}^D\chi_{12}\right)
-
V_1\left(T_{\rm h }^D\chi_{21}-T_{\rm c}^D\chi_{11}\right)
\right]
}{
T_{\rm h}^D\left(V_2\xi_{22}-V_1\xi_{21}\right)
}\,.
\label{7.33}
\end{equation}
In the flat-space limit, $\xi_{ij}\to1$ and $\chi_{ij}\to1$, so the efficiency simplifies to
\begin{equation}
\eta_{\rm regen}^{\rm CFT}
=
\frac{1}{D}\left(1-\left(\frac{T_{\rm c}}{T_{\rm h}}\right)^D\right)\,.
\label{7.34}
\end{equation}

Thus, unlike the classical ideal gas, a thermal CFT does not generically admit perfect
regeneration, and even in the flat-space limit the regenerative Stirling efficiency does
not coincide with the Carnot value.

This can be understood from the criterion derived in section~\ref{sec:stirlingtheorem}. For a thermal CFT,   one finds from
\eqref{7.13} in the high-temperature or large-volume expansion  that
\begin{equation}
C_V = T\left(\frac{\partial S}{\partial T}\right)_V \propto V T^{D-1} + \cdots\,,\end{equation}
so that $C_V$ depends explicitly on the volume. This follows from scale invariance, since $C_V$ can only depend on the temperature via the dimensionless product $TR$. \textcolor{black}{Volume dependence by itself only shows that the sufficient theorem does not
apply; the explicit calculation of the isochoric heat exchanges establishes
that the two isochores do not cancel and that the regenerative efficiency
remains below Carnot.}

Finally, by plotting the efficiencies  on a 3D plot for $\mathcal{N} = 4$ SYM theory at zero 't Hooft coupling (i.e. $a_4/a_2 = -1/6$) and at infinite 't Hooft coupling (i.e. $a_4/a_2 = -1/12$), both in the presence and absence of regeneration, we find: \begin{equation}\eta^{\text{regen}}_{\lambda =0} > \eta^{\text{regen}}_{\lambda \rightarrow \infty} > \eta^{\text{non-regen}}_{\lambda =0} > \eta^{\text{non-regen}}_{\lambda \rightarrow \infty}\,.\end{equation} From this, we deduce that regeneration and weak coupling enhance the efficiency of the CFT Stirling engine.

\section{Stirling engine for   holographic CFTs}
\label{sec:holographiccfts}
\subsection{Thermal state dual to AdS-Schwarzschild black hole}

We now turn to thermal CFT states that admit a holographic dual description in terms of
AdS  black holes. This provides an exact realization of the thermal CFT
thermodynamics discussed in the previous section, allowing us to go beyond the
high-temperature or large-volume expansion. In particular, the holographic dictionary  determines the
functions $S(T,V)$ and $P(T,V)$ exactly for specific AdS black holes, and hence gives an exact
expression for the Stirling efficiency.

We consider static, uncharged, asymptotically AdS black holes in $D+1$ bulk
spacetime dimensions with spherical, planar, or hyperbolic horizon topology,
which we collectively refer to as AdS-Schwarzschild black holes, even though
the horizon topology may be non-spherical. The dual CFT lives on a spatial manifold of constant curvature with radius $R$ and volume
\begin{equation} \label{volumecft}
V=\Omega_{k,D-1}R^{D-1},
\end{equation}
where $\Omega_{k,D-1}$ denotes the dimensionless volume of the unit constant-curvature
space. The parameter $k=1,0,-1$ corresponds to spherical, planar, and hyperbolic
geometry, respectively.   In this section we keep $k$ general. As in our earlier work \cite{Lilani:2025gzt}, we restrict to black holes with positive
heat capacity, so that the cycle operates as a heat engine rather than a refrigerator.

We parametrize the AdS-Schwarzschild solution in terms of the dimensionless horizon radius $x \equiv r_h / R, $
with $r_h$ the horizon radius. The thermodynamic
variables of the dual CFT are then given by \cite{Visser:2021eqk,Ahmed:2023snm}
\begin{align}
S &= 4\pi C x^{D-1},\qquad C=\frac{\Omega_{k,D-1}L^{D-1}}{16\pi G}\,, \label{Sblackhole}\\
E &= \frac{(D-1)C x^{D-2}}{R}\left(k+x^2\right),\\
T &= \frac{D-2}{4\pi R x}\left(k+\frac{D}{D-2}x^2\right),\\
P &= \frac{C x^{D-2}}{\Omega_{k,D-1}R^D}\left(k+x^2\right)\,. \label{Pblackhole}
\end{align}
Here, $C$ is the central charge of the dual CFT and $L$ is the AdS curvature radius. It is straightforward to check that these quantities satisfy the conformal equation of state \eqref{7.1}.
As in the previous section, the dependence on $R$ is fixed by conformal invariance, so
that $ER$ and $TR$ are dimensionless. \\

\noindent \textbf{In absence of regeneration.} For a generic CFT the   efficiency of a Stirling engine without regenerator is given by \eqref{7.10}. For CFT thermal states dual to AdS-Schwarzschild we can express the entropy and pressure exactly in terms of  $(T,V)$. We first solve the temperature relation for $x$:
\begin{equation}
Dx^2-4\pi TR\,x + (D-2)k =0 \,.
\end{equation}
Choosing the positive branch gives
\begin{equation}
x(T,V)
=
\frac{2\pi TR}{D}
\left[
1+\sqrt{
1-\frac{kD(D-2)}{4\pi^2T^2R^2}
}
\right]\,.
\label{xexpression}
\end{equation}
  This large black hole solution has positive heat capacity, as opposed to the small black hole solution, with a negative sign in front of the square root, which has a negative heat capacity. Note that for planar and hyperbolic neutral black holes there are no small black hole solutions with positive horizon radius. 
Substituting $R=R(V)$  from \eqref{volumecft}, and proceeding   as in the previous
section, the non-regenerative Stirling efficiency again takes the form
\begin{equation}
\eta_{\text{non-regen}}^{\text{CFT}}
=
1-\frac{
T_{\rm c}^{D}(V_{2}\xi_{12}-V_{1}\xi_{11})
+\frac{D-1}{D}V_{2}(T_{\rm h}^{D}\chi_{22}-T_{\rm c}^{D}\chi_{12})
}{
T_{\rm h}^{D}(V_{2}\xi_{22}-V_{1}\xi_{21})
+\frac{D-1}{D}V_{1}(T_{\rm h }^{D}\chi_{21}-T_{\rm c}^{D}\chi_{11})
}\,,
\label{eq:holo_nonregen_general_k}
\end{equation}
but now the functions $\xi_{ij}$ and $\chi_{ij}$ are known exactly. Keeping $k$ general,
they are
\begin{align}
\xi_{ij}
&=
\frac{1}{2^{D-1}}
\left[
1+\sqrt{
1-\frac{kD(D-2)}{4\pi^2T_i^2}
\left(\frac{\Omega_{k,D-1}}{V_j}\right)^{\frac{2}{D-1}}
}
\right]^{D-1}\,,
\label{eq:holo_xi_general_k}\\
\chi_{ij}
&=
\xi_{ij}^{\frac{D-2}{D-1}}
\left[
\xi_{ij}^{\frac{2}{D-1}}
+\frac{kD^2}{16\pi^2T_i^2}
\left(\frac{\Omega_{k,D-1}}{V_j}\right)^{\frac{2}{D-1}}
\right]\,.
\label{eq:holo_chi_general_k}
\end{align}
For $k=1$ these reduce to the spherical expressions derived in \cite{Lilani:2025gzt}, while for
$k=0$ one simply has $\xi_{ij}=1=\chi_{ij}$.\\

\noindent \textbf{In presence of regeneration.} We now turn to the new result, namely the Stirling efficiency in the presence of a
regenerator.  For a generic CFT the corresponding energy-balance expression is given by~\eqref{7.27b},  so we only need to insert the exact holographic functions $S(T,V)$ and $P(T,V)$ that follow from inserting \eqref{xexpression} into \eqref{Sblackhole} and \eqref{Pblackhole}.
This gives
\begin{equation}
\eta_{\text{regen}}^{\text{CFT}}
=
1-\frac{
T_{\rm c}^D(V_2\xi_{12}-V_1\xi_{11})
+\frac{D-1}{D}\left[
V_2(T_{\rm  h}^D\chi_{22}-T_{\rm c}^D\chi_{12})
-
V_1(T_{\rm h }^D\chi_{21}-T_{\rm c}^D\chi_{11})
\right]
}{
T_{\rm h }^D(V_2\xi_{22}-V_1\xi_{21})
}\,,
\label{eq:holo_regen_general_k}
\end{equation}
where the functions $\xi_{ij}$ and $\chi_{ij}$ are given above in
\eqref{eq:holo_xi_general_k} and \eqref{eq:holo_chi_general_k}.
\textcolor{black}{This is the exact regenerative energy-balance expression for a
holographic CFT state dual to an AdS-Schwarzschild black hole. As discussed in
section~\ref{sec:2.1}, it gives the reversible regenerative efficiency directly
when the local isochoric mismatch has the same sign throughout the temperature
interval.\footnote{For spherical horizons, $k=1$, the positive-heat-capacity
branch begins at $x^2=(D-2)/D$, whereas the black hole becomes thermodynamically
dominant only above the Hawking-Page transition at $x=1$. In the intermediate
subdominant region, $\sqrt{(D-2)/D}<x<1$, the heat capacity is non-monotonic and
the local isochoric mismatch can change sign. On the dominant branch $x\geq1$,
the heat capacity increases with $x$, so that for $V_2>V_1$ the local mismatch
has a fixed positive sign.}
The second term in the numerator represents the mismatch between the heat
released along $2\to3$ and that absorbed along $4\to1$. This mismatch is
generically nonzero, so the regenerator does not generally recycle all of
the isochoric heat, and the regenerative Stirling efficiency does not
generally reduce to the Carnot value.}

Finally, when $k=0$ the exact holographic expressions reduce to the planar thermal CFT
result \eqref{7.34}. For $k=\pm 1$, \eqref{eq:holo_xi_general_k} and \eqref{eq:holo_chi_general_k}
capture the exact finite-curvature corrections.

\subsection{Thermal state dual to AdS-Reissner-Nordstr\"{o}m black hole}
\label{subsec:stirling_equal_to_carnot}
Next, we consider a thermal, charged CFT state dual to a charged, static, spherically symmetric AdS black hole solution to Einstein-Maxwell  theory, i.e. an AdS-Reissner-Nordstr\"{o}m black hole. 
For these black holes the first law   contains an additional potential-charge term
\begin{equation}\label{firstlawenergywithcharge}
dE=\delta Q-PdV+\tilde\Phi\,d\tilde{\mathcal Q}\,.
\end{equation}
We denote boundary quantities with tildes, $\tilde{\mathcal{Q}}$ and $\tilde \Phi$, while the corresponding bulk quantities $\mathcal{Q}$ and $\Phi$ are without tildes, following the convention in \cite{Visser:2021eqk,Cong:2021jgb,Ahmed:2023snm}. 
We work below in the grand canonical ensemble, in which the temperature $T$, the volume $V$ and the electric potential $\tilde \Phi$ in the CFT are the independent variables. 
We will show   that, in this ensemble, the holographic Stirling engine can attain the Carnot efficiency    in the limit $\tilde \Phi \to \infty$, in contrast to the neutral case.  

\begin{figure}[tbp]
    \centering
    \begin{subfigure}[t]{0.45\textwidth}
        \centering
        \includegraphics[width=\textwidth]{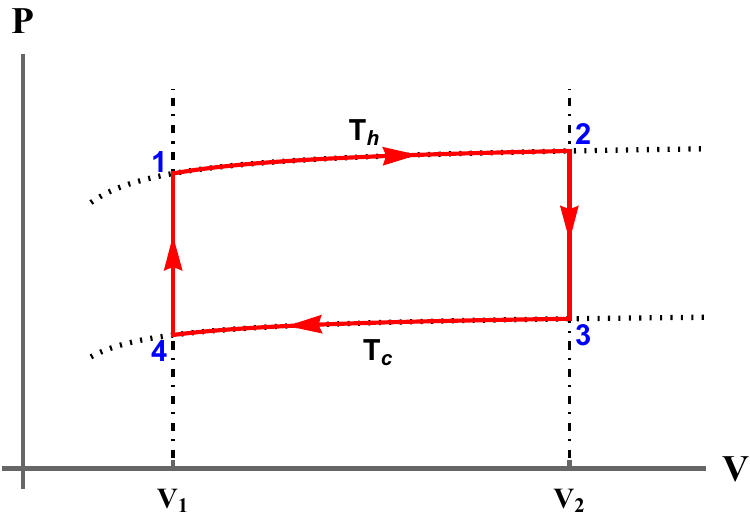} 
        \caption{Stirling cycle for an uncharged holographic CFT, where the   working substance is dual to an AdS-Schwarzschild black hole.}
        \label{fig:sub1}
    \end{subfigure}
    \hfill
    \begin{subfigure}[t]{0.45\textwidth}
        \centering
        \includegraphics[width=\textwidth]{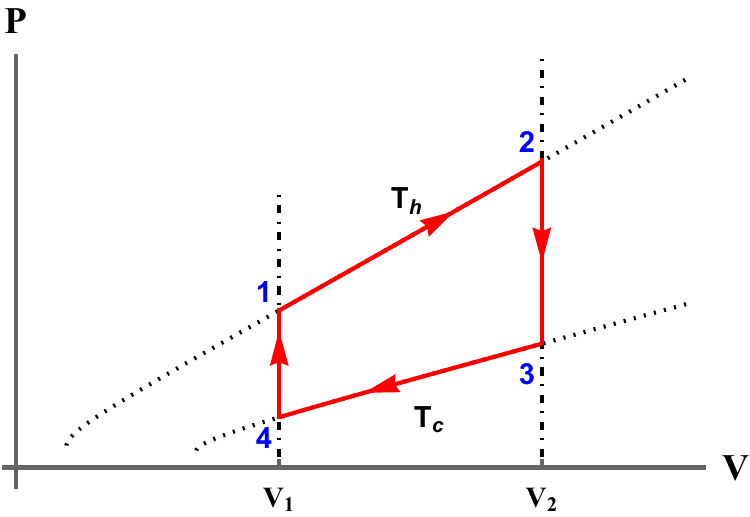} 
        \caption{Stirling cycle for a charged holographic CFT in the fixed $\tilde \Phi$ ensemble, where the   working substance is dual to an AdS-Reissner-Nordstr\"{o}m black hole.}
        \label{fig:sub2}
    \end{subfigure}

    \caption{$PV$-diagrams for holographic Stirling engines. $D=4$ and $k=1$ for these plots. }
    \label{fig:overall}
\end{figure}

\begin{figure}[t]
    \centering
    \includegraphics[width=0.5\textwidth]{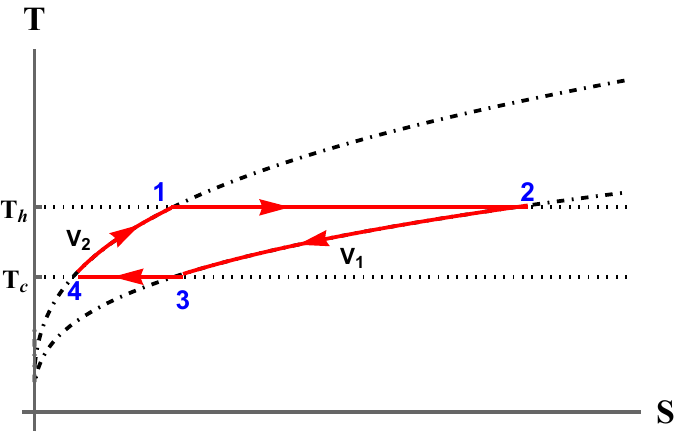}
    \caption{$TS$-diagram of a Stirling cycle for a charged holographic CFT in the fixed $\tilde{\Phi}$ ensemble. The corresponding diagram for an uncharged holographic CFT is qualitatively identical. $D=4$ and $k=1$ for this plot. }
    \label{}
\end{figure}

We first recall the holographic dictionary for a $(D+1)$-dimensional AdS-Reissner-Nordstr\"{o}m black hole \cite{Cong:2021jgb}. The CFT thermodynamic quantities are expressed in terms of the dimensionless parameter $x \equiv r_h/L$, the boundary curvature radius $R$, and the electric potential $\tilde \Phi$:
\begin{align}
\label{eqn:holo_cft_entro}
    S &= 4 \pi C x^{D-1}, \qquad  V = \Omega_{k,D-1} R^{D-1}\,,\\ 
\label{eqn:holo_cft_energy}
    E &= \frac{(D-1)Cx^{D-2}}{R}\left (k + x^{2} + \alpha^2 \tilde \Phi^2 R^2\right)\,,\\
\label{eqn:holo_temp} 
    T &= \frac{D-2}{4 \pi R x}\left (k + \frac{D}{D-2}x^{2} - \alpha^2 \tilde \Phi^2 R^2\right)\,,\\
     P &=\frac{C x^{D-2}}{\Omega_{k,D-1}R^D} \left(k + x^{2} + \alpha^2 \tilde \Phi^{2}R^{2} \right )\,,\\
       \tilde{\mathcal{Q}} &= 4 (D-2) C x^{D-2} \tilde \Phi R\,,
\end{align}
where 
\begin{equation}
    \alpha^2 = \frac{ 2(D-2)}{ D-1} \,,
\end{equation} $C$ is the central charge in \eqref{Sblackhole} and $k$ is the     horizon topology  parameter. These CFT quantities are related to the   black hole thermodynamic variables via
\begin{equation}
    S = \frac{A}{4G_N}\, \quad E = M \frac{L}{R}\,, \quad T = \frac{\kappa}{2\pi} \frac{L}{R}\,, \quad \tilde{\Phi} = \frac{\Phi}{R}\,, \quad \tilde{\mathcal{Q}} = \mathcal{Q} L\,,
\end{equation}
where $A$ is the horizon area, $M$ is the mass of the AdS-Reissner-Nordstr\"{o}m black hole, $\kappa$ is the surface gravity, $\Phi$ is the electric potential and $\mathcal Q$ is the electric charge.

Since we work in the grand canonical ensemble we need to solve for $x=x(T,R,\tilde \Phi)$. Choosing the large black hole with positive heat capacity yields
\begin{equation}
\label{eqn: x(T,R,phi)}
x =\frac{1}{D}\left ( 2 \pi   TR+\sqrt{  4 \pi^2 T^2 R^2 - D(D-2)(k - \alpha^2 \tilde \Phi^2 R^2)}\right)\,.
   \end{equation}
As expected for a CFT, this depends only on the scale-invariant combinations $TR$ and $\tilde{\Phi}R$. Substituting into the expressions above and using $R=R(V)$ yields all thermodynamic quantities as functions of $(T,V,\tilde\Phi)$.\\

\noindent \textbf{In absence of regeneration.} To compute the Stirling efficiency, we   need to understand the direction of heat flow along the various paths of the Stirling engine. Along the isotherm at higher temperature, the heat exchange is   $Q_{1 \rightarrow 2} = T_{\rm h}(S_2 - S_1)$. Since $S \propto x^{D-1}$, it is easy to see from \eqref{eqn: x(T,R,phi)} that $x_2 > x_1$ because $V_2 > V_1$. This implies that $Q_{1 \rightarrow 2} > 0$ and thus the heat flows into the system along the path $1 \rightarrow 2$. Similarly, since $x_3 > x_4$, $Q_{3 \rightarrow 4}= T_{\rm c} (S_4 - S_3)< 0$ and thus there is heat loss along the path $3 \rightarrow 4$. This implies that unlike in the fixed-charge ensemble, in the fixed-potential ensemble, there is no other possibility than the conventional heat exchange along the two isotherms. 

The heat exchange along the isochores is more interesting. From the conformal
equation of state~\eqref{7.1} and the first law~\eqref{firstlawenergywithcharge},
the heat exchange along an isochore is
\begin{equation}
Q=\int TdS
=(D-1)V\Delta P-\tilde{\Phi}\,\Delta\tilde{\mathcal Q}\,.
\end{equation}
Note that the charge is allowed to vary in the fixed-potential ensemble. Hence,
\begin{align}
Q_{2 \rightarrow3} &= (D-1)V_2(P_3 - P_2) - \tilde{\Phi} (\tilde{\mathcal{Q}}_3 - \tilde{\mathcal{Q}}_2)\,, \\
Q_{4\rightarrow1} &= (D-1)V_1(P_1 - P_4) - \tilde{\Phi} (\tilde{\mathcal{Q}}_1 - \tilde{\mathcal{Q}}_4)\,.
\end{align}
\textcolor{black}{Since we work on the branch with positive
$C_{V,\tilde\Phi}$, the direction of heat flow along the isochores follows
directly:
\begin{equation}
Q_{2\to3}
=
-\int_{T_{\rm c}}^{T_{\rm h}}
C_{V,\tilde\Phi}(T,V_2)\,dT<0,
\qquad
Q_{4\to1}
=
\int_{T_{\rm c}}^{T_{\rm h}}
C_{V,\tilde\Phi}(T,V_1)\,dT>0.
\end{equation}
Thus, heat is released during isochoric cooling and absorbed during
isochoric heating.}

Substituting $P(x,R,\tilde{\Phi})$ and $\tilde{\mathcal{Q}}(x,R, \tilde{\Phi})$
gives
\begin{equation}
    \label{eqn: q_23}
    Q_{2 \rightarrow 3} = \frac{(D-1)C}{R_2}  \left[(x_2^{D-2}-x_3^{D-2})(\alpha^2\tilde{\Phi}^{2}R_{2}^{2}-k)-(x_2^D-x_3^D)\right]\,,
\end{equation}
\begin{equation}
    \label{eqn: q_14}
    Q_{4 \rightarrow 1} = \frac{(D-1)C}{R_1}   \left[(x_1^D-x_4^D)-(x_1^{D-2}-x_4^{D-2})(\alpha^2\tilde{\Phi}^{2}R_{1}^{2}-k)\right]\,.
\end{equation}
In the conventional case $Q_{\rm in} = |Q_{1 \rightarrow 2}| + |Q_{4 \rightarrow 1}|$
and $Q_{\rm out} = |Q_{2 \rightarrow 3}| + |Q_{3 \rightarrow 4}|$. Substituting $x(T,R,\tilde{\Phi})$ into these expressions, we arrive at the following efficiency for the nonregenerative Stirling engine  
\begin{equation}
    \label{eff_in_out}
    \eta_{\rm non\text{-}regen}^{\rm CFT} = 1 - \frac{T_{\rm c}^D(V_2 \xi_{12} -V_1\xi_{11})+ V_2\frac{D-1}{D}(T_{\rm h}^D\chi_{22}-T_{\rm c}^D\chi_{12})}{T_{\rm h}^D(V_2 \xi_{22} -V_1\xi_{21})+ V_1\frac{D-1}{D}(T_{\rm h}^D\chi_{21}-T_{\rm c}^D\chi_{11})}\,,
\end{equation}
where $\xi_{ij}$ and $\chi_{ij}$ are  
\begin{align}
   \label{eqn: xi_charged_bh}
\xi_{ij}
&=
\frac{1}{2^{D-1}}
\left[
1+
\sqrt{
1+D(D-2)\left(
\frac{\alpha^2\tilde{\Phi}^{\,2}}{4\pi^2T_i^2}
-\frac{k}{4\pi^2T_i^2}
\left(\frac{\Omega_{k,D-1}}{V_j}\right)^{\!\frac{2}{D-1}}
\right)
}
\right]^{D-1},\\
\chi_{ij}
&=
\xi_{ij}^{\frac{D-2}{D-1}}
\left[
\xi_{ij}^{\frac{2}{D-1}}
+\frac{D^2}{16\pi^2T_i^2}
\left(
k\left(\frac{\Omega_{k,D-1}}{V_j}\right)^{\!\frac{2}{D-1}}
 -\alpha^2\tilde\Phi^{\,2}
\right)
\right].
   \label{eqn:chi_charged_bh}
\end{align}
For $\tilde \Phi=0$ these functions agree with our expressions in the case of AdS-Schwarzschild.\\

\noindent \textbf{In presence of regeneration.}
\textcolor{black}{For a regenerative Stirling engine in the fixed-$\tilde\Phi$
ensemble, the efficiency depends on the mismatch between the two isochoric
heat exchanges. In the fixed-sign regime relevant at sufficiently large
$\tilde\Phi$, the heat released during isochoric cooling exceeds that absorbed
during isochoric heating. In terms of the functions introduced above, this
condition is}
\begin{equation}
T_{\rm h}^D\bigl(V_2 \chi_{22}-V_1 \chi_{21}\bigr) 
> 
T_{\rm c}^D\bigl(V_2 \chi_{12}-V_1\chi_{11}\bigr)\,.
\end{equation}
As a consequence, the heat stored during the isochoric cooling step cannot be fully reutilized during the subsequent isochoric heating step. \textcolor{black}{For the process to remain reversible, the excess heat must be
rejected to external reservoirs at the temperatures at which it arises.
When the local mismatch has the same sign throughout the temperature interval,
the total excess heat is $Q_{\rm mis}$, so that}
$ 
Q_{\text{out}} = |Q_{3\to4}| + Q_{\text{mis}}.
$
Further, we saw that, in the fixed $\tilde{\Phi}$ ensemble, the heat exchange along the isotherms follows the conventional pattern. Therefore the regenerative energy-balance efficiency satisfies the same formula as for a generic CFT, cf. \eqref{eq:holo_regen_general_k}, 
\begin{equation} \label{regencftchargedeff}
\eta_{\rm regen}^{\rm CFT} 
= 1 - \frac{T_{\rm c}^D(V_2 \xi_{12} -V_1\xi_{11})+ \frac{D-1}{D}\left[V_2(T_{\rm h}^D\chi_{22}-T_{\rm c}^D\chi_{12})-V_1(T_{\rm h}^D\chi_{21}-T_{\rm c}^D\chi_{11})\right]}{T_{\rm h}^D(V_2 \xi_{22} -V_1\xi_{21})}\,,
\end{equation}
but now   $\xi_{ij}$ and $\chi_{ij}$ are given exactly by the expressions \eqref{eqn: xi_charged_bh} and \eqref{eqn:chi_charged_bh}  above.


Finally, we turn to the main result of this section. Unlike classical
working substances such as ideal gases or Van der Waals fluids, a charged
holographic CFT does not exhibit perfect cancellation of heat along the
isochores, so that $Q_{\rm mis}\neq0$. Nevertheless, both the regenerative
and non-regenerative efficiencies approach the Carnot efficiency in the
infinite potential limit, i.e.,
\begin{equation}
\eta_{\rm (non\text{-})regen}^{\rm CFT}
\to
\eta_{\rm Carnot}
\qquad\text{as}\qquad
\tilde\Phi\to\infty\,.
\end{equation}
\textcolor{black}{As we show below, both efficiencies have corrections of order
$\mathcal O(\tilde\Phi^{-1})$, but the coefficient of the leading correction
is smaller in the presence of regeneration. Regeneration therefore reduces
the leading deviation from the Carnot value. This is also shown in Figure~\ref{fig:eta_vs_phi}.}

\begin{figure}[t]
    \centering
    \includegraphics[width=0.5\textwidth]{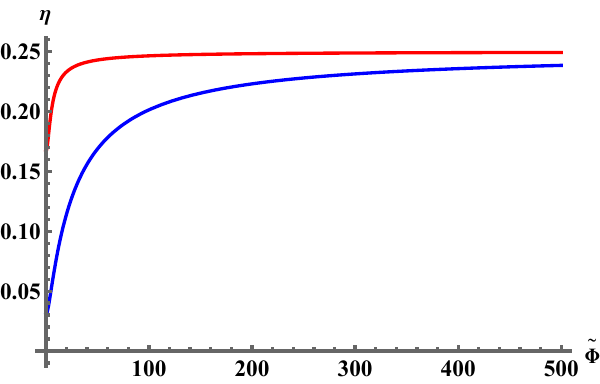}
    \caption{Stirling efficiency as a function of the electric potential  for charged holographic CFTs, at fixed $T_{\rm h}$, $T_{\rm c}$, $V_2$, and $V_1$. Here $T_{\rm h} = 2$, $T_{\rm c}=1.5$, $V_2=1$, $V_1=0.9$ and $D=4$. The red and blue lines correspond to the presence and absence of regeneration, respectively.  }
    \label{fig:eta_vs_phi}
\end{figure}

\textcolor{black}{We now prove these statements analytically. It is useful to define
\begin{equation}
\beta_D\equiv
\alpha\sqrt{\frac{D-2}{D}}\,,
\end{equation}
together with
\begin{equation}
\mathcal B_D\equiv
\frac{4\pi C}{\Omega_{k,D-1}}\beta_D^{D-1}\,,
\qquad
\mathcal A_D\equiv
\frac{8\pi^2C(D-1)}{D\,\Omega_{k,D-1}}\beta_D^{D-2}\,.
\end{equation}
At fixed $T$ and $R$, the large-potential expansion of
$x(T,R,\tilde\Phi)$ is
\begin{equation}
x
=
R\left(
\beta_D\tilde\Phi+\frac{2\pi T}{D}
+\mathcal O(\tilde\Phi^{-1})
\right)\,.
\end{equation}
Using $S=4\pi Cx^{D-1}$ and
$V=\Omega_{k,D-1}R^{D-1}$, it follows that
\begin{equation}\label{entropy_large_phi}
S(T,V,\tilde\Phi)
=
\mathcal B_D V\tilde\Phi^{D-1}
+\mathcal A_D TV\tilde\Phi^{D-2}
+\mathcal O(\tilde\Phi^{D-3})\,.
\end{equation}
The heat exchanges along the isotherms are therefore
\begin{align}
Q_{1\to2}
&=
\mathcal B_D T_{\rm h}(V_2-V_1)\tilde\Phi^{D-1}
+\mathcal A_D T_{\rm h}^2(V_2-V_1)\tilde\Phi^{D-2}
+\mathcal O(\tilde\Phi^{D-3})\,,
\\
\left|Q_{3\to4}\right|
&=
\mathcal B_D T_{\rm c}(V_2-V_1)\tilde\Phi^{D-1}
+\mathcal A_D T_{\rm c}^2(V_2-V_1)\tilde\Phi^{D-2}
+\mathcal O(\tilde\Phi^{D-3})\,.
\end{align}
Equation~\eqref{entropy_large_phi} also gives
\begin{equation} \label{newheatcaphere}
C_{V,\tilde\Phi}(T,V)
=
\mathcal A_D TV\tilde\Phi^{D-2}
+\mathcal O(\tilde\Phi^{D-3})\,,
\end{equation}
and hence the isochoric heat exchanges are
\begin{align}
Q_{\rm in}^{4\to1}
&=
\frac{\mathcal A_D}{2}V_1
\left(T_{\rm h}^2-T_{\rm c}^2\right)
\tilde\Phi^{D-2}
+\mathcal O(\tilde\Phi^{D-3})\,,
\\
Q_{\rm out}^{2\to3}
&=
\frac{\mathcal A_D}{2}V_2
\left(T_{\rm h}^2-T_{\rm c}^2\right)
\tilde\Phi^{D-2}
+\mathcal O(\tilde\Phi^{D-3})\,.
\end{align}
Since $V_2>V_1$, the heat released during isochoric cooling is larger
than that absorbed during isochoric heating. Thus,
\begin{equation} \label{eq:Qmis_large_phi}
Q_{\rm mis}
=
\frac{\mathcal A_D}{2}(V_2-V_1)
\left(T_{\rm h}^2-T_{\rm c}^2\right)
\tilde\Phi^{D-2}
+\mathcal O(\tilde\Phi^{D-3})\,.
\end{equation}
Substituting these expansions into
\eqref{eff_in_out} and \eqref{regencftchargedeff} gives
\begin{align}
\eta_{\rm regen}^{\rm CFT}
&=
1-\frac{T_{\rm c}}{T_{\rm h}}
-\frac{\Gamma_{\rm regen}}{\tilde\Phi}
+\mathcal O(\tilde\Phi^{-2})\,,
\\
\eta_{\rm non\text{-}regen}^{\rm CFT}
&=
1-\frac{T_{\rm c}}{T_{\rm h}}
-\frac{\Gamma_{\rm non\text{-}regen}}{\tilde\Phi}
+\mathcal O(\tilde\Phi^{-2})\,,
\end{align}
where
\begin{equation}
\Gamma_{\rm regen}
=
\frac{\mathcal A_D}{2\mathcal B_D}
\frac{(T_{\rm h}-T_{\rm c})^2}{T_{\rm h}}\,,
\quad \text{and}\quad
\Gamma_{\rm non\text{-}regen}
=
\Gamma_{\rm regen}
\left[
1+
\frac{V_1}{V_2-V_1}
\frac{T_{\rm h}+T_{\rm c}}{T_{\rm h}}
\right]\,.
\end{equation}
Consequently,
\begin{equation}
\Gamma_{\rm non\text{-}regen}-\Gamma_{\rm regen}
=
\frac{\mathcal A_D}{2\mathcal B_D}
\frac{V_1}{V_2-V_1}
\frac{(T_{\rm h}-T_{\rm c})^2(T_{\rm h}+T_{\rm c})}
{T_{\rm h}^2}
>0\,.
\end{equation}
Thus, both efficiencies approach the Carnot value with corrections of
order $\mathcal O(\tilde\Phi^{-1})$, but the leading deviation is smaller
for the regenerative engine.}

\textcolor{black}{
For completeness, the exact heat capacity in the fixed-potential ensemble
is
\begin{equation}
\begin{aligned}
C_{V,\tilde\Phi}
=
T\left(\frac{\partial S}{\partial T}\right)_{V,\tilde\Phi}
&=
4\pi C(D-1)x^{D-1}
\frac{(D-2)\bigl(k-\alpha^2\tilde\Phi^2R^2\bigr)+Dx^2}
{Dx^2-(D-2)\bigl(k-\alpha^2\tilde\Phi^2R^2\bigr)}
\\
&=
\frac{8\pi^2C(D-1)TR\,x^{D-1}}
{Dx-2\pi TR}\,.
\end{aligned}
\end{equation}
Its large-potential expansion agrees with the result \eqref{newheatcaphere} used above. The heat capacity does not become independent of the volume in the infinite-$\tilde \Phi$ limit; 
rather, it scales linearly with $V$.}


\textcolor{black}{
Further, the entropy mismatch between the two isochoric processes is
\begin{equation}
\begin{aligned}
\Delta S_{\rm mis}
&\equiv (S_2-S_3)-(S_1-S_4) \\
&=
\int_{T_{\rm c}}^{T_{\rm h}}
\frac{
C_{V,\tilde\Phi}(T,V_2)
-
C_{V,\tilde\Phi}(T,V_1)
}{T}\,dT \\
&=
\mathcal A_D(V_2-V_1)(T_{\rm h}-T_{\rm c})
\tilde\Phi^{D-2}
+\mathcal O(\tilde\Phi^{D-3})\,.
\end{aligned}
\end{equation}
This entropy mismatch and the heat mismatch $Q_{\rm mis}$ arise from the
same local difference between the two isochoric heat exchanges. At large
$\tilde\Phi$, the heat released during isochoric cooling exceeds that
required during isochoric heating at every temperature, so the local
mismatch does not cancel upon integration. The corresponding total surplus
heat is therefore the $Q_{\rm mis}$ obtained in
equation~\eqref{eq:Qmis_large_phi}.
For the regenerator to return reversibly to its initial state, this surplus
heat must be rejected externally at the temperatures at which it arises.
The quantity $\Delta S_{\rm mis}$ gives the associated entropy balance, but
does not introduce any additional heat beyond the $Q_{\rm mis}$ already
included in~\eqref{regencftchargedeff}.}

\textcolor{black}{Since $Q_{\rm mis}$ is of order $\tilde\Phi^{D-2}$, whereas the leading
isothermal heat is of order $\tilde\Phi^{D-1}$, the mismatch produces only a
relative correction of order $\mathcal O(\tilde\Phi^{-1})$. It therefore does
not change the limiting efficiency:
 $ 
\eta\to 1- T_{\rm c}/T_{\rm h}
.
$
Although $C_{V,\tilde\Phi}$ remains proportional to $V$ in the
large-$\tilde\Phi$ limit, this does not contradict the theorem in
section~\ref{sec:stirlingtheorem}. That theorem assumes fixed conserved
charges, whereas here $\tilde\Phi$ is fixed and the charge varies along
the cycle.}



\section{Discussion and Outlook}

In this work, we have been interested in constructing heat engines that make use of black hole thermodynamics via the holographic correspondence. Such constructions are particularly useful for strongly coupled systems, where thermodynamic quantities are not accessible from field theory and must instead be computed using the dual gravitational description. In this sense, our approach is similar in spirit to the use of holography in
identifying universal properties of strongly coupled systems, such as the shear
viscosity to entropy density ratio, as illustrated by the   KSS bound
$\eta/s \geq 1/4\pi$ \cite{Kovtun:2004de}. At the same time, the engines considered here rely on a number of idealizations, and it is therefore important to examine how the results are modified once these assumptions are relaxed.

An important extension of our analysis is to relax the assumption of
reversibility. In realistic engines, heat exchange between the working
substance and the reservoirs occurs across finite temperature differences,
which necessarily leads to entropy production and dissipation. A simple
framework to incorporate such effects is endoreversible thermodynamics \cite{Lavenda}, in
which the working substance remains in local equilibrium while irreversibility
is attributed to the heat transfer between the working substance and the reservoirs. This
shifts the emphasis from reversible efficiency to performance at finite time,
in particular to efficiency at maximum power, as in the Curzon-Ahlborn
efficiency \cite{CurzonAhlborn}.

Another direction is to relax the assumption of an ideal regenerator. Realistic regenerators are characterized by an effectiveness $\varepsilon_R < 1$, so that only a fraction of the heat exchanged during the isochoric processes is recovered. This leads to additional heat exchange with the external reservoirs and a corresponding reduction in efficiency. It would be useful to incorporate such effects into the present framework and to disentangle the contributions to the heat flows arising from non-ideal regeneration from the intrinsic mismatch~$Q_{\mathrm{mis}}$.

A further extension is to broaden the class of working substances. On the gravitational side, one can include rotation or higher-curvature corrections, the latter corresponding to $1/N$ effects in the dual field theory. It is also useful to consider different thermodynamic ensembles. In particular, charged black holes in the fixed charge ensemble provide a counterpart to the fixed-potential ensemble studied here.

On the field theory side, one may consider different geometries and scaling symmetries. In the present work, spherical working substances arise from thermal states dual to spherical black holes. A direct generalization is to consider compact working substances  with different shapes, such as a cube on the plane, which allows one to isolate finite-volume effects beyond the spherical case. It is also of interest to extend the analysis to holographic field theories with non-relativistic scaling symmetries, such as Lifshitz scaling, where the thermodynamic relations are modified.

Furthermore, it would be interesting to construct holographic heat engines for
dual thermal systems defined on finite timelike boundaries of black hole
spacetimes, rather than for CFTs at asymptotic infinity. In particular, the
recent holographic definition~\cite{Borsboom:2026ash} of pressure and volume
provides a framework for a broad class of black holes, with asymptotically flat
Schwarzschild as the simplest example, in which the dual theory lives on the
boundary. In this setting, the thermodynamic state space is non-degenerate, so
that entropy and volume are independent variables. This allows one to define
holographic heat engines beyond the AdS/CFT setting, and may provide new insight
into how the black hole geometry evolves during a thermodynamic cycle.

\textcolor{black}{Finally, our results suggest a connection between the
isochoric heat mismatch and phase structure. For the ideal Bose gas, the
direction of the mismatch reverses between the normal and Bose-Einstein
condensed phases: above the critical temperature,
$Q_{\rm out}^{2\to3}<Q_{\rm in}^{4\to1}$, whereas below it,
$Q_{\rm out}^{2\to3}>Q_{\rm in}^{4\to1}$. Equivalently, the signed
difference $Q_{\rm out}^{2\to3}-Q_{\rm in}^{4\to1}$ changes sign across
the transition.   It would be interesting to
investigate whether similar behavior occurs in other systems with phase
transitions.}

More broadly, these de-idealizations may also make other thermodynamic cycles more interesting. In previous work, cycles such as the Otto, Diesel, Brayton, and rectangular $PV$-cycles have efficiencies that are the same for all CFTs, being fixed entirely by the equation of state. Once the idealizations of our setup are relaxed, this universality  may no longer   hold, and these cycles may acquire a nontrivial dependence on the thermodynamic data of the system, providing additional probes of strongly coupled systems via holography.\\

\noindent \textbf{Acknowledgements.} MRV thanks the audiences of
the 2025 Peyresq Spacetime Meeting, the National Seminar Theoretical High Energy Physics
at NIKHEF, the IMAPP meeting at Radboud University and the physics colloquium at Bilkent University, where this work was presented,
for their interesting questions. He is also grateful to Bilkent University, and in particular to Bayram Tekin, for their generous hospitality at the final stages of this project. MRV is supported by the Spinoza Grant of the Dutch
Science Organization (NWO) awarded to Klaas Landsman. NL thanks the Oxford and Cambridge Society of India (OCSI) for their generous scholarship support and Hughes Hall, Cambridge for the financial support.

\appendix

\section{Isotherm and isochore equations}
In this appendix, we briefly describe how the $PV$- and $TS$-diagrams are
obtained for the various working substances (see also Appendix B in \cite{Lilani:2025gzt}). To construct the $P V$-cycle  for the Stirling engine it is necessary to
determine the equation for  an isotherm for each working substance. In the $PV$-plane the equation for an isochore is simply $V=\text{constant}.$ For the $TS$-cycle we need the equation for an isochore for each working substance. In this case the equation for an isotherm is, of course, $T=\text{constant}.$ We list the isotherm and isochore equations below for the relevant working substances.\\

\noindent \textbf{Isotherm equations for  various working substances in the $PV$-plane.} 
For a classical ideal gas, the equation for an isotherm at fixed temperature $T_o$ is
\begin{equation}
    P=\frac{NT_o}{V} \,.
\end{equation}
For the Van der Waals fluid 
\begin{equation}
    P = \frac{NT_o}{V-bN} - \frac{a N^2}{V^2} \,.
\end{equation}
For an ideal Bose gas
\begin{equation}
    P = \frac{T_o}{\lambda_{T_{o}}^d} g_{\frac{d}{2}+1}\left[z(V,T_o)\right] \,.
\end{equation}
For an ideal Fermi gas
\begin{equation}
    P = \frac{T_o}{\lambda_{T_{o}}^d} f_{\frac{d}{2}+1}\left[z(V,T_o)\right]\,.
\end{equation}
For a Bose-Einstein condensate
\begin{equation}
    P = A T_o^{d/2+1} = \text{constant}\,,
\end{equation}
where $A=\left(\frac{2\pi m}{h^2}\right)^{d/2}\zeta\!\left(\frac{d}{2} +1\right).$

For a holographic CFT thermal state dual to an AdS-Schwarzschild black hole (see equation (B2) in \cite{Lilani:2025gzt} for a derivation)
\begin{align}
\label{eqn: CFT_PV_isotherm}
    &PV^{\frac{D}{D-1}} = C (\Omega_{k,D-1})^{\frac{1}{D-1}} 
    \left[\frac{1}{D} \left(2 \pi T_o \left(\frac{V}{\Omega_{k,D-1}}\right)^{\frac{1}{D-1}} 
    + \sqrt{ 4 \pi ^{2} T_o^2\left(\frac{V}{\Omega_{k,D-1}}\right)^{\frac{2}{D-1}}- kD(D-2)}\right)\right]^{D-2}\!\!\!\!\!\!\!\!\times \\
    &\quad 
    \times\left[
   k + 
     \left(\frac{1}{D} \left(2 \pi T_o  \left(\frac{V}{\Omega_{k,D-1}}\right)^{\frac{1}{D-1}} + \sqrt{  4 \pi ^{2} T_o^2 \left(\frac{V}{\Omega_{k,D-1}}\right)^{ \frac{2}{D-1}}- kD(D-2)}\right)\right)^{2} 
    \right]\,.\nonumber
\end{align}
For a holographic CFT thermal state dual to an AdS-Reissner-Nordstr\"{o}m  black hole
\begin{align}
\label{eqn: CFT_PV_isotherm}
\hspace{-1cm}
\resizebox{1.2\textwidth}{!}{$
\begin{aligned}
&PV^{\frac{D}{D-1}} = C (\Omega_{k,D-1})^{\frac{1}{D-1}} 
\left[\frac{1}{D} \left(2 \pi T_o \left(\frac{V}{\Omega_{k,D-1}}\right)^{\frac{1}{D-1}} 
+ \sqrt{ 4 \pi ^{2} T_o^2\left(\frac{V}{\Omega_{k,D-1}}\right)^{\frac{2}{D-1}}- D(D-2)\left(k -\alpha^2 \tilde{\Phi}^2 \left(\frac{V}{\Omega_{k,D-1}}\right)^{\frac{2}{D-1}}\right)}\right)\right]^{D-2}\!\!\!\!\!\!\!\!\times \\
&\quad
\times\left[
 k  + 
 \left(\frac{1}{D} \left(2 \pi T_o  \left(\frac{V}{\Omega_{k,D-1}}\right)^{\frac{1}{D-1}} + \sqrt{  4 \pi ^{2} T_o^2 \left(\frac{V}{\Omega_{k,D-1}}\right)^{ \frac{2}{D-1}}- D(D-2)\left(k -\alpha^2 \tilde{\Phi}^2 \left(\frac{V}{\Omega_{k,D-1}}\right)^{\frac{2}{D-1}}\right)}\right)\right)^{2} 
 +\alpha^2 \tilde{\Phi}^2 \left(\frac{V}{\Omega_{k,D-1}}\right)^{\frac{2}{D-1}} 
\right]\,.
\end{aligned}
$}
\end{align}
\textbf{Isochore equations for various  working substances in the $TS$-plane.} 
For a classical ideal gas the equation for an isochore at fixed volume $V_o$ is
\begin{equation}
    \label{eqn:TS_IG_Isochore}
    T = \exp\left[\frac{(\gamma - 1)S}{ N} - \gamma \right] \left(\frac{V_o}{N}\right)^{-(\gamma - 1)}\frac{1}{2\pi m}  \,,
\end{equation}
where $m$ is the mass of the gas particle and $\gamma \equiv C_P/C_V$ is the ratio of heat capacities, which is related to the number of degrees of freedom $f$ by $\gamma = 1+2/f$ for an ideal gas.

\noindent For the Van der Waals fluid
\begin{equation}
    \label{eqn:TS_IG_Isochore}
    T = \exp\left[\frac{(\gamma - 1)S}{ N} - \gamma \right] \left(\frac{V_o-bN}{N}\right)^{-(\gamma - 1)}\frac{1}{2\pi m}  \,.
\end{equation}
For a Bose-Einstein condensate
\begin{equation}
    \label{eqn:BEC_isochore}
    T = \left(\frac{S}{A(d/2+1)V_o}\right)^{\frac{2}{d}}\,.
\end{equation}
For  a holographic CFT thermal state dual to an AdS-Schwarzschild black hole
\begin{equation}
    \label{TS_CFT_Isochore}
     T = \frac{D-2}{4 \pi}\left(\frac{\Omega_{k,D-1}}{V_o}\right)^{\frac{1}{D-1}}\left(\frac{4 \pi C}{S}\right)^{\frac{1}{D-1}}\left ( k  + \frac{D}{D-2}\left(\frac{S}{4 \pi C}\right)^{\frac{2}{D-1}} \right)
 \,.\end{equation}
For a  holographic CFT thermal state dual to an AdS-Reissner-Nordstr\"{o}m black hole
\begin{equation}
    \label{TS_CFT_Isochore}
     T = \frac{D-2}{4 \pi}\left(\frac{\Omega_{k,D-1}}{V_o}\right)^{\frac{1}{D-1}}\left(\frac{4 \pi C}{S}\right)^{\frac{1}{D-1}}\left ( k + \frac{D}{D-2}\left(\frac{S}{4 \pi C}\right)^{\frac{2}{D-1}}-\alpha^2 \tilde{\Phi}^2 \left(\frac{V_o}{\Omega_{k,D-1}}\right)^{\frac{2}{D-1}}\right)
 \,.\end{equation}

\bibliographystyle{JHEP}
\bibliography{refs}

\end{document}